%% file: chsch.tex
\documentclass[letterpaper,11pt]{article}
\usepackage[T1]{fontenc}
\usepackage{myfullpage}
\usepackage{amsmath,amsthm}
\usepackage{amssymb,latexsym}
\usepackage{psfrag,graphicx}
\usepackage{mathrsfs}
\usepackage{xspace}
\usepackage[usenames]{color}

\usepackage[unicode,final=true]{hyperref}

\newtheorem{theorem}{Theorem}
\newtheorem{corollary}[theorem]{Corollary}

\newtheorem{lemma}[theorem]{Lemma}

\newcounter{pln}
\newenvironment{smallenum}
{\begin{list} {(\roman{pln})} { \usecounter{pln}%
     \settowidth{\labelwidth}{(iii)}%
     \setlength{\leftmargin}{\labelwidth}%
     \addtolength{\leftmargin}{1.5 \labelsep}%
     \setlength{\itemsep}{0cm}%
     \setlength{\parsep}{0cm}%
     \setlength{\topsep}{0mm} }
}
{\end{list}}

\newcommand{\RR}{\mathbb{E}}
\newcommand{\NN}{\mathbb{N}}

\newcommand{\cC}{\mathcal{C}}

\newcommand{\fF}{\mathcal{F}}

\newcommand{\lL}{\mathcal{L}}
\newcommand{\pP}{\mathcal{P}}
\newcommand{\qQ}{\mathcal{Q}}

\newcommand{\sS}{\mathcal{S}}

\newcommand{\PPP}{\mathscr{P}}

\renewcommand{\to}{\rightarrow}

\newcommand{\nin}{\not\in}

\newcommand{\sm}{{\setminus}}
\newcommand{\polar}[1]{{#1}^\star}
\newcommand{\polarhp}[1]{{#1}^{\star-}}
\newcommand{\lexp}[1]{\lfloor\frac{#1}{2}\rfloor}
\newcommand{\cexp}[1]{\lceil\frac{#1}{2}\rceil}

\newcommand{\etal}{\textit{et al.}\xspace}
\newcommand{\resp}{resp.,\xspace}
\newcommand{\ie}{i.e.,\xspace}

\newcommand{\sep}{,\xspace}

\begin{document}
\title{Convex hulls of spheres and convex hulls of convex
  polytopes lying on parallel hyperplanes}

\author{Menelaos I. Karavelas$^{\dagger,\ddagger}$\hfil{}
Eleni Tzanaki$^{\dagger,\ddagger}$\\[5pt]
\it{}$^\dagger$Department of Applied Mathematics,\\
\it{}University of Crete\\
\it{}GR-714 09 Heraklion, Greece\\
{\small\texttt{\{mkaravel,etzanaki\}@tem.uoc.gr}}\\[5pt]
\it{}$^\ddagger$Institute of Applied and Computational Mathematics,\\
\it{}Foundation for Research and Technology - Hellas,\\
\it{}P.O. Box 1385, GR-711 10 Heraklion, Greece}

\maketitle

\begin{abstract}
  Given a set $\Sigma$ of spheres in $\mathbb{E}^d$, with $d\ge{}3$
  and $d$ odd, having a fixed number of $m$ distinct radii
  $\rho_1,\rho_2,\ldots,\rho_m$, we show that the worst-case
  combinatorial complexity of the convex hull $CH_d(\Sigma)$ of
  $\Sigma$ is
  $\Theta(\sum_{1\le{}i\ne{}j\le{}m}n_in_j^{\lfloor\frac{d}{2}\rfloor})$,
  where $n_i$ is the number of spheres in $\Sigma$ with radius $\rho_i$.
  Our bound refines the worst-case upper and lower bounds on the
  worst-case combinatorial complexity of $CH_d(\Sigma)$ for all odd
  $d\ge{}3$.

  To prove the lower bound, we construct a set of $\Theta(n_1+n_2)$
  spheres in $\mathbb{E}^d$, with $d\ge{}3$ odd, where $n_i$ spheres
  have radius $\rho_i$, $i=1,2$, and $\rho_2\ne\rho_1$, such that
  their convex hull has combinatorial complexity
  $\Omega(n_1n_2^{\lfloor\frac{d}{2}\rfloor}+n_2n_1^{\lfloor\frac{d}{2}\rfloor})$.
  Our construction is then generalized to the case where the spheres
  have $m\ge{}3$ distinct radii.

  For the upper bound, we reduce the sphere convex hull
  problem to the problem of computing the worst-case combinatorial
  complexity of the convex hull of a set of $m$ $d$-dimensional
  convex polytopes lying on $m$ parallel hyperplanes in
  $\mathbb{E}^{d+1}$, where $d\ge{}3$ odd, a problem which is of
  independent interest.
  More precisely, we show that the worst-case combinatorial complexity
  of the convex hull of a set
  $\{\mathcal{P}_1,\mathcal{P}_2,\ldots,\mathcal{P}_m\}$
  of $m$ $d$-dimensional convex polytopes lying on $m$ parallel
  hyperplanes of $\mathbb{E}^{d+1}$ is
  $O(\sum_{1\le{}i\ne{}j\le{}m}n_in_j^{\lfloor\frac{d}{2}\rfloor})$,
  where $n_i$ is the number of vertices of $\mathcal{P}_i$.
  This bound is an improvement over the worst-case
  bound on the combinatorial complexity of the convex hull of a point
  set where we impose no restriction on the points'
  configuration; using the lower bound construction for the sphere
  convex hull problem, it is also shown to be tight for all odd
  $d\ge{}3$.

  Finally: (1) we briefly discuss how to compute convex hulls of
  spheres with a fixed number of distinct radii, or convex hulls of a
  fixed number of polytopes lying on parallel hyperplanes; (2) we show
  how our tight bounds for the parallel polytope convex hull problem,
  yield tight bounds on the combinatorial complexity of the Minkowski
  sum of two convex polytopes in $\mathbb{E}^d$; and (3) we state some
  open problems and directions for future work.

  \bigskip\noindent
  \textit{Key\;words:}\/
  high-dimensional geometry\sep discrete geometry\sep
  combinatorial geometry\sep combinatorial complexity\sep
  convex hull\sep Minkowski sum\sep spheres\sep convex
  polytopes\sep parallel hyperplanes
  \smallskip\par\noindent
\end{abstract}

\clearpage

\input{intro}
\input{chp}
\input{chp-algo}
\input{chs}
\input{chs-algo}

\input{concl}

\section*{Acknowledgments}
Partially supported by the FP7-REGPOT-2009-1 project ``Archimedes
Center for Modeling, Analysis and Computation''.

\bibliographystyle{plain}
\bibliography{chsch}

\end{document}

%% file: intro.tex

\section{Introduction and results}
\label{sec:intro}

Let $\Sigma$ be a set of $n$ spheres in $\RR^d$, $d\ge{}2$, where the
dimension $d$ is fixed. We call $\Pi$ a \emph{supporting} 
hyperplane of $\Sigma$ if it has non-empty intersection with $\Sigma$
and $\Sigma$ is contained in one of the two closed halfspaces bounded
by $\Pi$. We call $H$ a \emph{supporting halfspace} of the set
$\Sigma$ if it contains all spheres in $\Sigma$ and is limited by
a supporting hyperplane $\Pi$ of $\Sigma$. The intersection of all
supporting halfspaces of $\Sigma$ is called the convex hull
$CH_d(\Sigma)$ of $\Sigma$. The definition of convex hulls detailed
above is applicable not only to spheres, but also to any finite
set of compact geometric objects in $\RR^d$. In the case of points,
i.e., if we have a set $P$ of $n$ points in $\RR^d$, the worst-case
combinatorial complexity\footnote{In the rest of the paper, and unless
  otherwise stated, we use the term ``complexity'' to
  refer to ``combinatorial complexity''.} of $CH_d(P)$ is known to be
$\Theta(n^{\lexp{d}})$. Moreover, there exist worst-case optimal
algorithms for constructing $CH_d(P)$ that run in
$O(n^{\lexp{d}}+n\log{}n)$ time, e.g., see
\cite{g-eadch-72,ph-chfsp-77,a-aeach-79,p-ortap-79,c-ochaa-93}.
Since the complexity of $CH_d(P)$ may
vary from $O(1)$ to $\Theta(n^{\lexp{d}})$, a lot of work has been
devoted to the design of output-sensitive algorithms for constructing
$CH_d(P)$, i.e., algorithms the running time of which depends on the
size of the output convex hull $CH_d(P)$, e.g., see
\cite{ck-acp-70,j-ichfs-73,s-chdch-86,ks-upcha-86,ms-loq-92,
cm-dosch-95,ar-cvddp-96,c-oosch-96,c-osrch-96,csy-pddpo-97}.
For a nice overview of the various algorithms for computing the convex
hull of points sets, the interested reader may refer to the paper by
Erickson \cite{e-nlbch-99}, while Avis, Bremner and Seidel
\cite{abs-hgach-97} have a very nice discussion about the
effectiveness of output-sensitive convex hull algorithms for point
sets.

Results about the convex hull of non-linear objects are very
limited. Aurenhammer \cite{a-pdpaa-87} showed that the worst-case
complexity of the power diagram of a set of $n$
spheres in $\RR^d$, $d\ge{}2$, is $O(n^{\cexp{d}})$, which also
implies the same upper bound for the worst-case
complexity for the convex hull of the sphere set.
Rappaport \cite{r-chada-92} devised an $O(n\log{}n)$
algorithm for computing the convex hull of a set of discs on the
plane, which is worst-case optimal. Boissonnat \etal
\cite{bcddy-acchs-96} give an $O(n^{\cexp{d}} + n\log{}n)$
algorithm for computing the convex hull of a set of $n$ spheres
in $\RR^d$, $d\ge{}2$, which is worst-case optimal in three and
even dimensions, since they also show that the worst-case
complexity of the convex hull of $n$ spheres in
$\RR^3$ is $\Theta(n^2)$. Finally, their results hold true for the
case of homothetic convex objects. Boissonnat and Karavelas
\cite{bk-ccevc-03} settled a conjecture in \cite{bcddy-acchs-96}: they
proved that the worst-case complexity of a set of $n$
spheres in $\RR^d$, $d\ge{}2$, is $\Theta(n^{\cexp{d}})$, which
also implied that the algorithm presented in \cite{bcddy-acchs-96} is
optimal for all $d$. As far as output-sensitive algorithms are
concerned, Boissonnat, C{\'e}r{\'e}zo and Duquesne \cite{bcq-osacc-92}
showed how to construct the convex hull of a set of $n$
three-dimensional spheres in $O(nf)$ time, where $f$ is the size of
the output convex hull, while Nielsen and Yvinec \cite{ny-oscha-98}
discuss optimal or almost optimal output-sensitive convex hull
algorithms for planar convex objects.

In this paper we consider the problem of computing the
complexity of the convex hull of a set of spheres, when the spheres
have a fixed number of distinct radii. This problem has been posed by
Boissonnat and Karavelas \cite{bk-ccevc-03}, and it is meaningful for
odd dimensions only: in even dimensions the complexity of both the
convex hull of $n$ points and the convex hull of $n$ spheres is
$\Theta(n^{\lexp{d}})=\Theta(n^{\cexp{d}})$, \ie the
two bounds match.

Consider a set of $n$ spheres $\Sigma$ in $\RR^d$, where $d\ge{}3$ and
$d$ odd, such that the spheres in $\Sigma$ have a fixed number $m$ of
distinct radii $\rho_1,\rho_2,\ldots,\rho_m$. Let $n_i$ be the
number of spheres in $\Sigma$ with radius $\rho_i$.
We say that $\rho_\lambda$ \emph{dominates} $\Sigma$ if
$n_\lambda=\Theta(n)$. 
We further say that $\Sigma$ is \emph{uniquely (\resp strongly)
  dominated}, if, for some $\lambda$, $\rho_\lambda$ dominates
$\Sigma$, and $n_i=o(n)$ (\resp $n_i=O(1)$), for all $i\ne\lambda$.
Using this terminology, we can qualitatively express our results as
follows.
Firstly, if $\Sigma$ is strongly dominated, then, from the
combinatorial complexity point of view, $CH_d(\Sigma)$ behaves as if
we had a set of points, or equivalently a set of spheres with the same
radius. If, however, $\Sigma$ is dominated by at least two radii,
$CH_d(\Sigma)$ behaves as in the generic case, where we impose no
restriction on the number of distinct radii in $\Sigma$. Finally, if
$\Sigma$ is uniquely dominated (but not strongly dominated), the
complexity of $CH_d(\Sigma)$ stands in-between the two extremes above:
the complexity of $CH_d(\Sigma)$ is asymptotically larger than the
points' case (or when we have spheres with the same radius), and
asymptotically smaller than the generic case, where we impose no
restriction on the number of distinct radii in $\Sigma$.
From the quantitative point of view, our results refine the results in
\cite{bk-ccevc-03} for any odd dimension $d\ge{}3$ as follows.
We prove that the worst-case complexity of $CH_d(\Sigma)$ is
$\Theta(\sum_{1\le{}i\ne{}j\le{}m}n_in_j^{\lexp{d}})$. 
This tight bound constitutes an improvement over the generic worst-case
complexity of $CH_d(\Sigma)$ if $\Sigma$ is uniquely dominated, since
in this case the complexity of $CH_d(\Sigma)$ is $o(n^{\cexp{d}})$. On
the other hard, it matches the generic worst-case complexity of
$CH_d(\Sigma)$, if $\Sigma$ is dominated by at least two radii: in
this case the worst-case complexity of $CH_d(\Sigma)$ becomes
$\Theta(n^{\cexp{d}})$. Finally, if $\Sigma$ is strongly dominated,
the complexity of $CH_d(\Sigma)$ is $\Theta(n^{\lexp{d}})$, \ie
it matches the worst-case complexity of convex hulls of point sets (or
sets of spheres where all spheres have the same radius).

To establish the lower bound for the complexity of $CH_d(\Sigma)$, we
construct a set $\Sigma$ of $\Theta(n_1+n_2)$ spheres in $\RR^d$, for
any odd $d\ge{}3$, where $n_1$ spheres have radius $\rho_1$ and $n_2$
spheres have radius $\rho_2\ne\rho_1$, such that worst-case complexity
of $CH_d(\Sigma)$ is
$\Omega(n_1n_2^{\lexp{d}}+n_2n_1^{\lexp{d}})$. This construction is
then generalized to sets of spheres having a fixed number of $m\ge{}3$
distinct radii. More precisely, we construct a set $\Sigma$ of
$n=\sum_{i=1}^mn_i$ spheres, where $n_i$ spheres have radius
$\rho_i$, with the $\rho_i$'s being pairwise distinct, such that the
worst-case complexity of $CH_d(\Sigma)$ is
$\Omega(\sum_{1\le{}i\ne{}j\le{}m}n_in_j^{\lexp{d}})$.

To prove our upper bound we use a lifting map, introduced in
\cite{bcddy-acchs-96}, that lifts spheres $\sigma_i=(c_i,r_i)$ in
$\RR^d$ to points $p_i=(c_i,r_i)$ in $\RR^{d+1}$. The convex hull 
$CH_d(\Sigma)$ is then the intersection of the hyperplane $\{x_{d+1}=0\}$ 
with the Minkowski sum of the convex hull $CH_{d+1}(P)$ and the hypercone
$\lambda_0$, where $P$ is the point set $\{p_1,p_2,\ldots,p_n\}$ in
$\RR^{d+1}$, and $\lambda_0$ is the lower half hypercone with
arbitrary apex, vertical axis and angle at the apex equal to
$\frac{\pi}{4}$.
When the spheres in $\Sigma$ have a fixed number $m$ of distinct
radii, the points of $P$ lie on $m$ hyperplanes parallel to
the hyperplane $\{x_{d+1}=0\}$. In this setting, computing the complexity
of $CH_d(\Sigma)$ amounts to computing the complexity of the convex hull
of $m$ convex polytopes lying on $m$ parallel hyperplanes of
$\RR^{d+1}$. This observation gives rise to the second major result
in this paper, which is of independent interest, and gives as
corollary a tight bound on the worst-case complexity of the Minkowski
sum of two convex $d$-polytopes.
Given a set $\PPP=\{\pP_1,\pP_2,\ldots,\pP_m\}$ of $m$ convex
$d$-polytopes in $\RR^{d+1}$, with $d\ge{}3$ and $d$ odd, we
show that the worst-case complexity of the convex hull
$CH_{d+1}(\PPP)$ is $\Theta(\sum_{1\le{}i\ne{}j\le{}m}n_in_j^{\lexp{d}})$,
where $n_i$ is the number of vertices of $\pP_i$.
Our upper bound proof is by induction on the number $m$ of parallel
hyperplanes. The lower bound follows from the lower bound on the
complexity of the convex hull of spheres having $m$ distinct radii.
Our bound constitutes an improvement over the worst-case complexity of
convex hulls of points sets, if a single polytope of $\PPP$ has
$\Theta(n)$ vertices, whereas all other polytopes have $o(n)$
vertices, where $n$ is the total number of vertices of all $m$
polytopes, while it matches the worst-case complexity of convex hulls
of points sets if at least two polytopes have $\Theta(n)$ vertices.

The rest of our paper is structured as follows:
In Section \ref{sec:chp} we detail our inductive proof of the upper
bound on the worst-case complexity of the convex hull of
convex polytopes lying on parallel hyperplanes, while in Section
\ref{sec:chp-algo} we discuss how to compute this convex hull.
In Section \ref{sec:chs} we prove our upper bound on the worst-case
complexity of the convex hull of a set of spheres.
Next we present our lower bound construction for any odd $d\ge{}3$
in two steps: first for sphere sets with two distinct radii
and then for sphere sets with $m\ge{}3$ distinct radii. We end the
section by discussing how this lower bound yields a tight lower bound
for the problem of the previous section.
In Section \ref{sec:chs-algo} we explain how to modify the algorithm
by Boissonnat \etal \cite{bcddy-acchs-96} so as to almost optimally
compute the convex hull of a set of spheres with a fixed number of
distinct radii.
Finally, in Section \ref{sec:concl} we summarize our results, 
we explain how our results yield tight bounds for the complexity of
the Minkowski sum of two convex polytopes, and state some open
problems.

%% file: chp.tex

\setlength\abovedisplayskip{7pt plus 3pt minus 7pt}
\setlength\belowdisplayskip{7pt plus 3pt minus 7pt}

\section{Convex hulls of convex polytopes lying on parallel
  hyperplanes}
\label{sec:chp}

A \emph{convex polytope}, or simply \emph{polytope}, $\pP$ in
$\RR^d$ is the convex hull of a finite set of points $P$ in
$\RR^d$.
A polytope $\pP$ can equivalently be described as the intersection of
all the closed halfspaces containing $P$. 
A \emph{face} of $\pP$ is an intersection of $\pP$ with hyperplanes
for which the polytope is contained in one of the two closed
halfspaces determined by the hyperplane.
The dimension of a face of $\pP$ is the dimension of its affine hull.
A $k$-face of $\pP$ is a $k$-dimensional face of $\pP$.
We consider the polytope itself as a trivial $d$-dimensional face; all
the other faces are called \emph{proper} faces. We will use the term
\emph{$d$-polytope} to refer to a polytope the trivial face of which
is $d$-dimensional.
For a $d$-polytope $\pP$, the $0$-faces of $\pP$ are its
\emph{vertices}, the $1$-faces of $\pP$ are its \emph{edges}, the
$(d-2)$-faces of $\pP$ are called \emph{ridges}, while the
$(d-1)$-faces are called \emph{facets}.
For $0\le{}k{}\le{}d$ we denote by $f_k(\pP)$ the number of $k$-faces
of $\pP$.
Note that every $k$-face $F$ of $\pP$ is also a $k$-polytope whose
faces are all the faces of $\pP$ contained in $F$.
A $k$-simplex in $\RR^d$, $k\le{}d$, is the convex hull of any
$k+1$ affinely independent points in $\RR^d$.
A polytope is called \emph{simplicial} if all its proper faces are
simplices. Equivalently, $\pP$ is simplicial if for every vertex $v$
of $\pP$ and every face $F\in\pP$, $v$ does not belong to the affine
hull of the vertices in $F\sm\{v\}$.

A \emph{polytopal complex} $\cC$ is a finite collection of polytopes
in $\RR^d$ such that
(i) $\emptyset\in\cC$,
(ii) if $\pP\in\cC$ then all the faces of $\pP$ are also in $\cC$ and
(iii) the intersection $\pP\cap\qQ$ for two polytopes in $\cC$ is a
face of both $\pP$ and $\qQ$. 
The dimension $\dim(\cC)$ of $\cC$ is the largest dimension of a
polytope in $\cC$.
A polytopal complex is called \emph{pure} if all its maximal (with
respect to inclusion) faces have the same dimension. 
In this case the maximal faces are called the \emph{facets} of
$\cC$. We will use the term \emph{$d$-complex} to refer to a pure
polytopal complex whose facets are $d$-dimensional.
A polytopal complex is simplicial if all its faces are simplices. 
Finally, a polytopal complex $\cC'$ is called a subcomplex
of a polytopal complex $\cC$ if all faces of $\cC'$ are also faces
of $\cC$.

One important class of polytopal complexes arise from polytopes.
More precisely, a $d$-polytope $\pP$, together with all its
faces and the empty set, form a polytopal $d$-complex, denoted
by $\cC(\pP)$. The only maximal face of $\cC(\pP)$, which is clearly
the only facet of $\cC(\pP)$, is the polytope $\pP$ itself.
Moreover, all proper faces of $\pP$ form a pure polytopal complex,
called the \emph{boundary complex} $\cC(\partial\pP)$. The facets of
$\cC(\partial\pP)$ are just the facets of $\pP$, and its dimension is
$\dim(\pP)-1=d-1$.
Given a polytope $\pP$ and a vertex $v$ of $\pP$, the \emph{star}
of $v$ is the polytopal complex of all faces of $\pP$ that contain
$v$, and their faces. The \emph{link} of $v$ is the subcomplex of the
star of $v$ consisting of all the faces of the star of $v$ that do not
contain $v$.

The $f$-vector $(f_{-1}(\pP),f_0(\pP),\ldots,f_{d-1}(\pP))$
of a simplicial $d$-polytope $\pP$ is defined as the $(d+1)$-dimensional
vector consisting of the number $f_k(\pP)$ of $k$-faces of $\pP$,
$-1\le{}k\le{}d$, where $f_{-1}(\pP)=1$ refers to the empty set.
The $h$-vector $(h_0(\pP),h_1(\pP),\ldots,h_d(\pP))$
of a simplicial $d$-polytope $\pP$ is defined as the
$(d+1)$-dimensional vector, where
$h_k(\pP):=\sum_{i=0}^{k}(-1)^{k-i}\binom{d-i}{d-k}f_{i-1}(\pP)$,
$0\le{}k\le{}d$.
The number $h_k(\pP)$ counts the number of facets of $\pP$ in a
\emph{shelling} of $\pP$, whose \emph{restriction} has size $k$; this
number is independent of the particular shelling chosen
(cf. \cite[Theorem 8.19]{z-lp-95}). It is easy to verify from the
defining equations of the $h_k(\pP)$'s that the elements of the
$f$-vector determine the elements of the $h$-vector and
vice versa. Moreover, the elements of the $f$-vector (or,
equivalently, the $h$-vector) are not linearly independent; they
satisfy the so called \emph{Dehn-Sommerville equations}, which can be
written in a very concise form as:
$h_k(\pP)=h_{d-k}(\pP)$, $0\le{}k\le{}d$.
An important implication of the existence of the Dehn-Sommerville
equations is that if we know the face numbers $f_k(\pP)$ for all
$0\le{}k\le\lexp{d}-1$, we can determine the remaining face
numbers $f_k(\pP)$ for all $\lexp{d}\le{}k\le{}d-1$.

In what follows we recall some facts from \cite[Section 5.2]{g-cp-03} that
will be of use to us later. Let $\pP$ be a $d$-polytope in $\RR^d$,
$F$ a facet of $\pP$, and $H$ the supporting hyperplane of $F$ (with
respect to $\pP$). For an arbitrary point $p$ in $\RR^d$, we say that
$p$ is \emph{beyond} (\resp \emph{beneath}) the facet $F$ of $\pP$, if
$p$ lies in the open halfspace of $H$ that does not contain $\pP$
(\resp contains the interior of $\pP$).
Furthermore, we say that an arbitrary point $v'$ is \emph{beyond} the
vertex $v$ of $\pP$ if for every facet $F$ of $\pP$ containing $v$,
$v'$ is beyond $F$, while for every facet $F$ of $\pP$ not containing
$v$, $v'$ is beneath $F$.
The vertices of the polytope $\pP'=CH_d((P\sm\{v\})\cup\{v'\})$ 
are the same with those of $\pP$, except for $v$ which has been
replaced by $v'$. In this case we say that $\pP'$ is obtained from
$\pP$ by \emph{pulling} $v$ to $v'$.
The vertex $v'$ of $\pP'$ does not belong to the affine hull of the
vertices in $F'\sm\{v'\}$ for every face $F'$ of $\pP'$.
The following result is well-known.

\begin{theorem}[\cite{k-nvcp-64,m-pkccp-70}]
  \label{McMullen}
  Let $\pP$ be a $d$-polytope. 
  \begin{smallenum}
  \item The $d$-polytope $\pP'$ we obtain by pulling a vertex of $\pP$
    has the same number of vertices with $\pP$, and
    $f_k(\pP)\le{}f_k(\pP')$ for all $1\le{}k\le{}d-1$.  
  \item The $d$-polytope $\pP'$ we obtain by successively pulling each
    of the vertices of $\pP$ is simplicial, has the same number of
    vertices with $\pP$, and $f_k(\pP)\le{}f_k(\pP')$ for all
    $1\le{}k\le{}d-1$.  
  \end{smallenum}
\end{theorem}

In the rest of the paper, when we refer to parallel hyperplanes we
assume that they have the same unit normal vector, \ie they have the
same orientation. Moreover, if two hyperplanes $\Pi$ and $\Pi'$ are
parallel, we say that $\Pi'$ is \emph{above} $\Pi$ if $\Pi'$ lies in
the positive open halfspace delimited by $\Pi$.

Let $\PPP=\{\pP_1,\pP_2,\ldots,\pP_m\}$ be a set of $m$
$d$-polytopes lying on $m$ parallel hyperplanes
$\Pi_1,\Pi_2,\ldots,$ $\Pi_m$ of $\RR^{d+1}$, respectively. Throughout
this section we assume that $m\ge{}2$ is fixed, and that $\Pi_j$ is
above $\Pi_i$ for all $j>i$. 
We denote by $P_i$ the set of vertices of $\pP_i$, by $n_i$ the
cardinality of $P_i$, and by $P$ the union
$P=P_1\cup{}P_2\cup\ldots\cup{}P_m$. Let $\pP=CH_{d+1}(P)$;
note that, for each $i$, not all vertices in $P_i$ are necessarily
vertices of $\pP$. Furthermore, among the polytopes in $\PPP$,
only $\pP_1$ and $\pP_m$ are faces of $\pP$.

The theorem that follows is the adaptation of Theorem
\ref{McMullen} in the context of $m$ parallel polytopes.
Again, we want to perturb
the points in $P$ so that $\pP'$ (the polytope we obtain after
perturbing the points in $P$) is simplicial with
$f_k(\pP)\le{}f_k(\pP')$ for $k\ge{}1$, but we want to retain the
property that the points lying on a hyperplane $\Pi_i$, if perturbed,
are replaced by points that lie on the same hyperplane.
This is almost possible. More precisely, all the faces of the polytope
$\pP'$ are simplicial, with the possible exception of
$\pP_1'=\pP'\cap\Pi_1$ and $\pP_m'=\pP'\cap\Pi_m$.

\begin{lemma}\label{simplicial}
  Let $\PPP=\{\pP_1,\pP_2,\ldots,\pP_m\}$ be a set of $m\ge{}2$
  $d$-polytopes lying on $m$ parallel hyperplanes
  $\Pi_1,\Pi_2,\ldots,\Pi_m$ of $\RR^{d+1}$, respectively, where
  $\Pi_j$ is above $\Pi_i$ for all $j>i$.
  Let $P_i$ be the vertex set of $\pP_i$, $1\le{}i\le{}m$, 
  $P=P_1\cup{}P_2\ldots\cup{}P_m$, and $\pP=CH_{d+1}(P)$.
  The points in $P$ can be perturbed in such a way that:
  \begin{smallenum}
  \item the points of $P$ in each hyperplane $\Pi_i$ remain in
    $\Pi_i$, $1\le{}i\le{}m$,
  \item all the faces of $\pP'$, except possibly the facets
    $\pP_1'$ and $\pP_m'$, are simplices, and,
  \item $f_k(\pP)\le{}f_k(\pP')$ for  all $1\le{}k\le{}d$, 
  \end{smallenum}
  where $\pP'$ is the polytope we obtain after having perturbed the
  vertices of $\pP$ in $P$.
\end{lemma}

\begin{proof}
  We construct $\pP'$ in three steps. Firstly, we properly perturb the
  points in $\Pi_1$ so that $\pP_1'$ is simplicial, then
  we do the same for the points in $\Pi_m$, and, finally, we pull
  every vertex of $\pP$ in $\Pi_i$ for $2\le{}i\le{}m-1$.

  Let $v\in\Pi_1$ and choose any $v'$ in $\Pi_1$ beyond all the
  facets of $\pP$ other than $\pP_1$ that contain $v$.
  If we focus our attention on the hyperplane $\Pi_1$ and on the
  polytope $\pP_1\in\Pi_1$ we have that $v$ is a vertex of
  $\pP_1$ and $v'$ is some point beyond all the facets of $\pP_1$
  containing $v$.  
  We consider the polytope $\pP'= CH_{d+1}((P\sm\{v\})\cup\{v'\})$.
  In view of \cite[Theorems 5.1.1 \& 5.1.2]{g-cp-03} the $k$-faces of
  $\pP'$ are either:
  \begin{smallenum}
  \item $k$-faces of $\pP_1'$, or
  \item $k$-faces of $\pP$ not in $\Pi_1$ that do not contain
    $v$, or
  \item faces of the form $CH_{d+1}(\{v'\}\cup G_{k-1})$, where
    $G_{k-1}$ is a $(k-1)$-face not containing $v$ of a
    facet $F\nin\Pi_1$ of $\pP$ containing $v$.
  \end{smallenum}
  The number of $k$-faces of $\pP'$ does not change due to Case
  (ii). The number of $k$-faces of $\pP'$ may increase in Case (iii)
  if $G_{k-1}$ is contained in a non-simplicial $k$-face of $\pP$,
  otherwise it does not increase. Finally, Case (i) is a consequence
  of Theorem \ref{McMullen}(i).

  We set  $\pP:=\pP'$ and we repeat the above procedure for every
  point $v\in{}P_1$. After having perturbed all points in 
  $P_1$ we obtain a polytope $\pP'$ such that $\pP_1'$
  is a simplicial polytope.
  Moreover, we have the additional properties that
  $f_k(\pP)\le{}f_k(\pP')$, for $k\ge{}1$, and that for every
  $v'\in\Pi_1$ and for each face $F'\in\pP'$ not lying in $\Pi_1$,
  the point $v'$ does not lie in the affine hull of the vertices in
  $F'\sm\{v'\}$.
  We repeat exactly the same procedure on the points of $\Pi_m$.
  Again, we denote by $\pP$ the polytope we obtain after having
  perturbed the points in $P_1$ and $P_m$.

  It is left to pull the points in the remaining hyperplanes;
  in fact, we need only perturb the points in $P$ that are vertices of
  $\pP$.
  Consider a vertex $v$ of $\pP$, such that $v\in{}P_i$ for
  some $2\le{}i\le{}m-1$. Let $v'$ be any point in $\Pi_i$ that lies
  beyond all the facets of $\pP$ containing $v$.
  The choice of such a $v'$ is possible since $\Pi_i$ is a hyperplane
  containing $v$ but not a supporting hyperplane of $\pP$.
  The polytope $\pP'=CH_{d+1}((P\sm\{v\})\cup\{v'\})$ is the one we
  obtain by pulling the vertex $v$ to $v'$ and thus from Theorem
  \ref{McMullen}(i) we have $f_k(\pP)\le{}f_k(\pP')$ for all $k\ge{}1$.
  We continue the same procedure for every vertex $v$ of $\pP$ in
  $\Pi_i$, and $2\le{}i\le{}m-1$.
  After having pulled all the vertices of $\pP$ in $\Pi_i$, for
  $2\le{}i\le{}m-1$, we obtain a polytope $\pP'$ with the property
  that for every vertex $v'$ of $\pP'$ in $\Pi_i$, with
  $2\le{}i\le{}m-1$, and for every face $F'\in\pP'$, the point $v'$
  does not lie in the affine hull of the vertices in $F'\sm\{v'\}$.

  Summarizing all of the above, we deduce that, after having perturbed
  all the vertices of $\pP$ in $P$, we get a polytope $\pP'$ with the
  same number of vertices as $\pP$, such that:
  (1) $f_k(\pP)\le{}f_k(\pP')$ for all $k\ge{}1$,
  (2) $\pP_1'$ and $\pP_m'$ are simplicial, and
  (3) for every vertex $v'$ of $\pP'$ and every face $F'\in\pP'$, such
  that $F'$ does not lie in $\Pi_1$ or $\Pi_m$, the vertex $v'$ does
  not lie in the affine hull of the vertices in $F'\sm\{v'\}$.
  The latter implies that all the faces of $\pP'$, not in $\Pi_1$ or
  $\Pi_m$, are simplices. This completes our proof.
\end{proof}

\definecolor{MyGreen}{rgb}{0,0.82,0}
\definecolor{MyBlue}{rgb}{0,0,1}
\definecolor{MyRed}{rgb}{1,0,0}
\definecolor{MyBrown}{rgb}{0.749,0.38,0}
\definecolor{MyMagenta}{rgb}{0.808,0,0.808}
\begin{figure}[t]
  \begin{center}
    \psfrag{P1}[][]{\textcolor{MyGreen}{\tiny$\pP_1$}}
    \psfrag{P2}[][]{\tiny$\pP_2$}
    \psfrag{P3}[][]{\tiny$\pP_3$}
    \psfrag{P4}[][]{\textcolor{MyRed}{\tiny$\pP_4$}}
    \psfrag{P~}[][]{\textcolor{MyBrown}{\tiny$\tilde{\pP}$}}
    \psfrag{Pi1}[][]{\textcolor{MyGreen}{\tiny$\Pi_1$}}
    \psfrag{Pi2}[][]{\tiny$\Pi_2$}
    \psfrag{Pi3}[][]{\tiny$\Pi_3$}
    \psfrag{Pi4}[][]{\textcolor{MyRed}{\tiny$\Pi_4$}}
    \psfrag{Pi~}[][]{\textcolor{MyBrown}{\tiny$\tilde{\Pi}$}}
    \psfrag{C}[][]{\textcolor{MyBlue}{\tiny$\fF$}}
    \psfrag{L}[][]{\textcolor{MyGreen}{\tiny$\lL$}}
    \psfrag{Fj}[][]{\textcolor{MyBlue}{\tiny$F_j$}}
    \psfrag{Fj-1}[][]{\textcolor{MyBlue}{\tiny$F_{j-1}$}}
    \psfrag{Fj~}[][]{\textcolor{MyBrown}{\tiny$\tilde{F}_j$}}
    \psfrag{Fj-1~}[][]{\textcolor{MyBrown}{\tiny$\tilde{F}_{j-1}$}}
    \includegraphics[width=\textwidth]{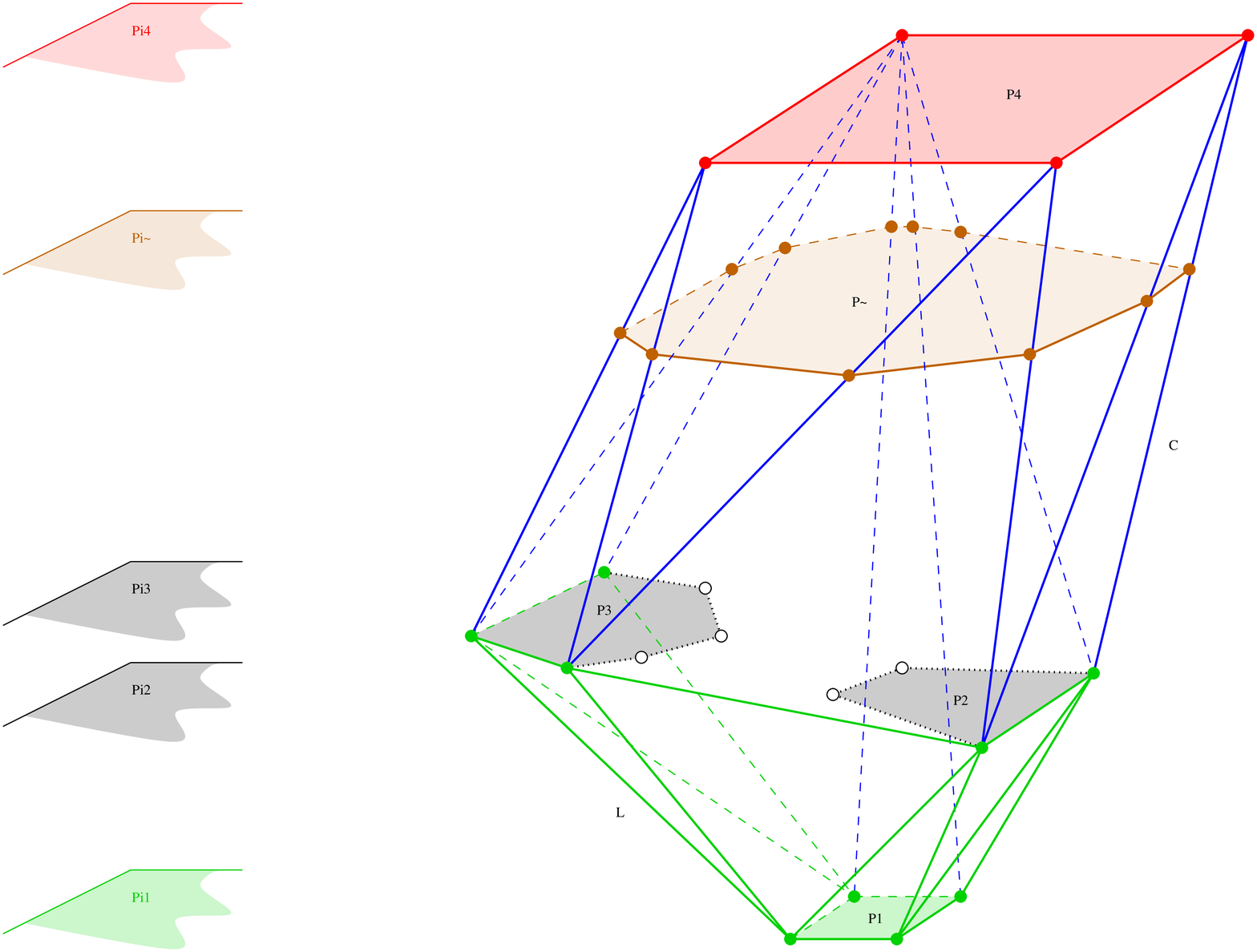}
  \end{center}
  \caption{An example of four polytopes lying on four parallel
    hyperplanes of $\RR^d$. The white vertices are not vertices of
    $\pP$. The polytopes $\pP_2$ and $\pP_3$ are shown in black. All
    faces of $\lL$ are shown in green (only the facet $\pP_1$ is
    shown), whereas all faces of $\pP_4$ are shown in red.
    The faces in blue are faces of $\fF$, whereas
    the faces in brown are faces of
    $\tilde{\pP}=\pP\cap\tilde{\Pi}$.}
  \label{fig:parallel_polytopes}
\end{figure}

In view of Lemma \ref{simplicial}, it suffices to restrict
our attention to sets of polytopes $\PPP$, where $\pP$ is simplicial
with the possible exception of its two facets $\pP_1$ and $\pP_m$.
Let $\tilde\Pi$ be any hyperplane between and parallel to the
hyperplanes $\Pi_{m-1} $ and $\Pi_m$ and consider the intersection
$\tilde\pP:=\pP\cap\tilde\Pi$ (see Fig. \ref{fig:parallel_polytopes}).
Let $\fF$ be the set of faces of $\pP$
having non-empty intersection with $\tilde{\Pi}$.
Note that $\tilde\pP$ is a $d$-polytope, which is, in general,
non-simplicial, and whose proper non-trivial faces
are intersections of the form $F\cap\tilde\Pi$ where $F\in\fF$.
Let $A=(\alpha_1,\alpha_2,\ldots,\alpha_m)$ and
$B=(\beta_1,\beta_2,\ldots,\beta_m)$ be two vectors in $\NN^m$. We
say that $A\preccurlyeq{}B$ if $\alpha_i\le\beta_i$ for all
$1\le{}i\le{}m$, and denote by $|A|$ the sum of the elements of
$A$, \ie $|A|=\sum_{i=1}^m\alpha_i$.
The following lemma provides an upper bound on the number of $k$-faces
of $\fF$.

\begin{lemma}\label{fboundhigh}
  The number of $k$-faces of $\fF$ is bounded from above as follows:
  \[   f_k(\fF)\le
  \sum_{\substack{(0,\ldots,0,1)\preccurlyeq{}A\preccurlyeq{}(k,\ldots,k)\\|A|=k+1}}
  \quad\prod_{i=1}^m\binom{f_0(\pP_i)}{\alpha_i},\qquad 1\le{}k\le{}d,
  \]
  where $\alpha_i$ is the $i$-th coordinate of the vector $A\in\NN^m$.
\end{lemma}

\begin{proof}
  According to Lemma \ref{simplicial}, it suffices to consider
  the case where $\pP$ is simplicial except possibly for its facets
  $\pP_1$ and $\pP_m$. In this context, a $k$-face $F\in\fF$ is
  simplicial and it is defined by $k+1$ vertices of $P$, 
  where at least one vertex comes from $P_m$, whereas the remaining
  $k$ vertices are vertices of $P_1,P_2,\ldots,P_{m-1}$. Let
  $\alpha_i$ be the number of vertices of $F$ from $P_i$. Clearly, we
  have $0\le\alpha_i\le{}k$ for $1\le{}i\le{}m-1$, $1\le\alpha_m\le{}k$,
  and $\alpha_1+\alpha_2+\ldots+\alpha_m=k+1$. The maximum number of
  possible $(\alpha_i-1)$-faces of $\pP_i$ is $\binom{f_0(\pP_i)}{\alpha_i}$,
  which implies that the maximum possible number of $k$-faces of $\fF$
  is $\prod_{i=1}^m\binom{f_0(\pP_i)}{\alpha_i}$. Summing over all possible
  values for the $\alpha_i$'s we get the desired expression.
\end{proof}


Using the bounds for the face numbers of $\fF$ from Lemma
\ref{fboundhigh}, we arrive at the following lemma concerning bounds
on the number of $k$-faces of $\pP$, for small values of $k$.

\begin{lemma}\label{lem:fk-bound-low}
  Let $n_i=f_0(\pP_i)$, $1\le{}i\le{}m$. The following asymptotic
  bounds hold:
  \begin{equation*}
    f_k(\pP)=O(\sum_{1\le{}i\ne{}j\le{}m}n_in_j^{k}
    +\sum_{i=1}^mn_i^{\min\{k+1,\lexp{d}\}}),\qquad 0\le{}k\le{}\textstyle\lexp{d}.
  \end{equation*}
\end{lemma}

\begin{proof}
  In what follows $k$ will be at most $\lexp{d}$, while we denote by
  $k'$ the quantity $k'=\min\{k+1,\lexp{d}\}$.
  Let $T(m)$ be the worst-case complexity of $f_k(\pP)$.
  We are going to prove, by induction on $m$, that for all $m\ge{}1$,
  $T(m)\le{}c\,(\sum_{1\le{}i\ne{}j\le{}m}n_in_j^{k}+\sum_{i=1}^mn_i^{k'})$,
  where $c$ is some appropriately large constant that depends only on $d$.

  The case $m=1$ is trivial since the number of $k$-faces of a
  $d$-polytope with $n_1$ vertices is in $O(n_1^{k'})$.
  Let us now assume that $m\ge{}2$ and that our statement holds for
  $m-1$; we shall prove it for $m$.
  To this end, we consider a set of $m$ $d$-polytopes
  $\PPP=\{\pP_1,\pP_2,\ldots,\pP_m\}$, lying on $m$ parallel
  hyperplanes $\Pi_1,\ldots,\Pi_m$ of $\RR^{d+1}$, such that $\Pi_j$
  is above $\Pi_i$ for all $j>i$.
  Let $n_i=f_0(\pP_i)$, $1\le{}i\le{}m$, and denote by $\pP$ the
  convex hull $CH_{d+1}(\PPP)$.
  Consider a hyperplane $\tilde\Pi$ parallel to $\Pi_{m-1}$ and $\Pi_m$
  and between them, and let $\tilde{\pP}=\pP\cap\tilde{\Pi}$ (refer to
  Fig. \ref{fig:parallel_polytopes}). We denote by $\lL$ the set
  of faces $F$ of $\pP$ such that $F\nin\pP_m$ and
  $F\cap\tilde{\Pi}=\emptyset$, and by $\fF$ the set of faces $F$ of
  $\pP$ with $F\cap\tilde{\Pi}\ne\emptyset$.
  Clearly, the set of faces of $\pP$ is equal to the disjoint union of
  $\lL$, $\fF$ and the set of faces of $\pP_m$; hence:
  \begin{equation}\label{equ:fk-decomp}
    f_k(\pP)=f_k(\lL)+f_k(\fF)+f_k(\pP_m).
  \end{equation}
  By the induction hypothesis we have that the number of $k$-faces of
  $CH_{d+1}(\PPP\sm\{\pP_m\})$ is at most $T(m-1)$. Since the
  $k$-faces in $\lL$ are $k$-faces of $CH_{d+1}(\PPP\sm\{\pP_m\})$, we
  have that $f_k(\lL)$ is at most $T(m-1)$.
  On the other hand, by Lemma \ref{fboundhigh}, we immediately get
  that $f_k(\fF)=O(n_m^k\sum_{i=1}^{m-1}n_i+n_m\sum_{i=1}^{m-1}n_i^k)$.
  Finally, since $\pP_m$ is a $d$-polytope, $f_k(\pP_m)=O(n_m^{k'})$.
  Combining these bounds with eq. \eqref{equ:fk-decomp}, we arrive at
  the following recurrence relation for $T(m)$:
  {
    \setlength\abovedisplayskip{4pt plus 3pt minus 7pt}
    \setlength\belowdisplayskip{4pt plus 3pt minus 7pt}
    \[ T(m)\le{}T(m-1)+
    O(n_m^{k}\sum_{i=1}^{m-1}n_i+n_m\sum_{i=1}^{m-1}n_i^{k})+O(n_m^{k'}).
    \]
  }%
  It is straightforward to verify that $T(m)$ satisfies:
  $T(m)\le{}c\,(\sum_{1\le{}i\ne{}j\le{}m}n_in_j^{k}
  +\sum_{i=1}^mn_i^{k'})$, for some appropriately large
  constant $c$ (that depends only on $d$); this establishes our claim.
\end{proof}

Exploiting the bounds from Lemma \ref{lem:fk-bound-low} for
$0\le{}k\le\lexp{d}$, and using the Dehn-Sommerville equations of an
appropriately defined simplicial $(d+1)$-polytope containing all
faces in $\pP$, except for the facets $\pP_1$ and $\pP_m$, we derive
asymptotic bounds on $f_k(\pP)$, for all $\lexp{d+1}\le{}k\le{}d$.
Our results are summarized in the following lemma.

\begin{lemma}\label{lem:fk-bound-high}
  Let $n_i=f_0(\pP_i)$, $1\le{}i\le{}m$. The following asymptotic
  bounds hold:
  \begin{equation*}
    f_k(\pP)=O(\sum_{1\le{}i\ne{}j\le{}m}n_in_j^{\lexp{d}}),
    \qquad {\textstyle\lexp{d+1}}\le{}k\le{}d.
  \end{equation*}
\end{lemma}

\begin{proof}
  Let $y$ (\resp $z$) be a point below $\Pi_1$ (\resp above $\Pi_m$),
  such that the vertices of $\pP_1$ (\resp $\pP_m$) are the only
  vertices of $\pP$ visible from $y$ (\resp $z$). To achieve this, we
  choose $y$ (\resp $z$) to be a point beyond the facet $\pP_1$ (\resp
  $\pP_m$) of $\pP$, and beneath every other facet of $\pP$.
  Let $Q$ be the set of points consisting of $y$, $z$ and the vertices
  of $\pP$, and let $\qQ=CH_{d+1}(Q)$ (refer to Fig. \ref{fig:Q}).
  Observe that the faces of $\pP$, except for the facets $\pP_1$ and
  $\pP_m$, are all faces of $\qQ$. To see that, notice that a
  supporting hyperplane $H_F$ for a facet $F\in\pP$, with
  $F\ne\pP_1,\pP_m$, is also a supporting hyperplane for
  $\qQ$. Indeed, the vertices of $F$ are vertices in
  $Q\setminus\{y,z\}\equiv{}P$, 
  %
  and thus, every vertex in $P$ that is not a vertex of $F$ strictly
  satisfies all hyperplane inequalities for $\pP$. Also,  by
  construction, the points $y$ and $z$ strictly satisfy all hyperplane
  inequalities apart from those for $\Pi_1$ and $\Pi_m$,
  respectively. Since $H_F$ is a hyperplane other than $\Pi_1$ and
  $\Pi_m$ we deduce that all points in $P\cup\{y,z\}$  lie on the same
  halfspace defined by $H_F$, and therefore $H_F$ supports $\qQ$.  
  The faces of $\qQ$
  that are not faces of $\pP$ are the faces in the star $\sS_y$ of $y$
  and the star $\sS_z$ of $z$. To verify this, consider a $k$-face $F$
  of $\pP_1$, and let $F_1$ be a face in $\pP_1$ that contains $F$. Let
  $H_1$ be a supporting hyperplane of $F_1$ with respect to
  $\pP$. Tilt $H_1$ until it hits the point $y$, while keeping $H_1$
  incident to $F_1$, and call $H_1'$ this tilted hyperplane. $H_1'$ is
  a supporting hyperplane for $y$ and the vertex set of $\pP_1$, and
  thus is a supporting hyperplane for $\qQ$. The same argument can be
  applied for the star of $z$. In fact, the boundary complex
  $\partial\pP_1$ (\resp $\partial\pP_m$) of $\pP_1$ (\resp $\pP_m$)
  is nothing but the link of $y$ (\resp $z$) in $\qQ$.

  \newcommand{\labelsize}{\small}
  \begin{figure}[p]
    \begin{center}
      \psfrag{P1}[][]{\textcolor{MyGreen}{\labelsize$\pP_1$}}
      \psfrag{P2}[][]{\labelsize$\pP_2$}
      \psfrag{P3}[][]{\labelsize$\pP_3$}
      \psfrag{P4}[][]{\textcolor{MyRed}{\labelsize$\pP_4$}}
      \psfrag{Pi1}[][]{\textcolor{MyGreen}{\labelsize$\Pi_1$}}
      \psfrag{Pi2}[][]{\labelsize$\Pi_2$}
      \psfrag{Pi3}[][]{\labelsize$\Pi_3$}
      \psfrag{Pi4}[][]{\textcolor{MyRed}{\labelsize$\Pi_4$}}
      \psfrag{F}[][]{\textcolor{MyBlue}{\labelsize$\fF$}}
      \psfrag{L}[][]{\textcolor{MyGreen}{\labelsize$\lL$}}
      \psfrag{y}[][]{\textcolor{MyMagenta}{\labelsize$y$}}
      \psfrag{z}[][]{\textcolor{MyMagenta}{\labelsize$z$}}
      \includegraphics[width=0.95\textwidth]{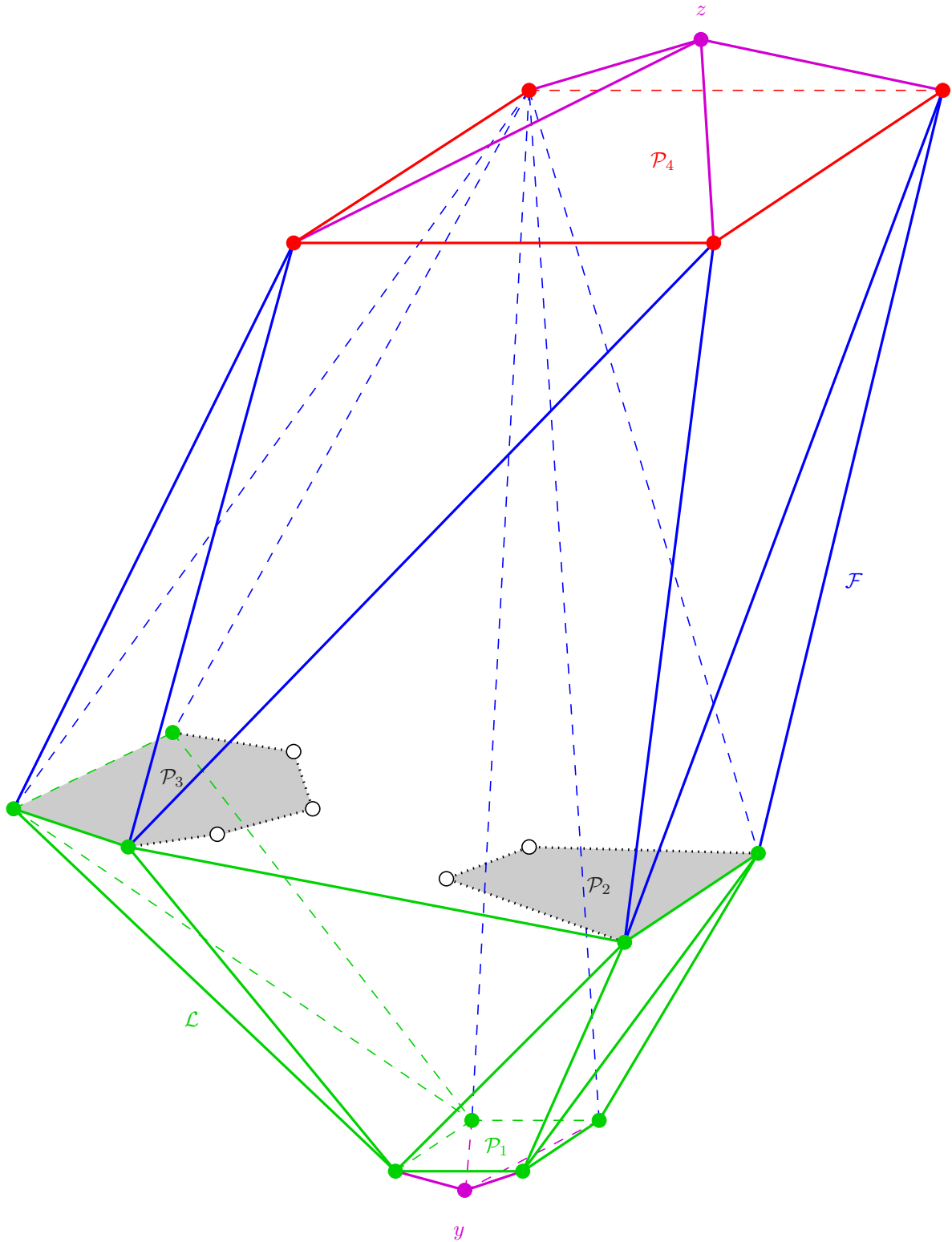}
    \end{center}
    \caption{The polytope $\qQ$. $\qQ$ is constructed by adding a
      point $y$ below $\Pi_1$ and a point $z$ above $\Pi_m$, such the
      only vertices of $\pP$ visible from $y$ (\resp $z$) are the
      vertices of $\pP_1$ (\resp $\pP_m$).}
    \label{fig:Q}
  \end{figure}

  It is easy to realize that, for $k<d$, the set of $k$-faces of $\qQ$
  is the union of the $k$-faces of $\pP$, $\sS_y$ and $\sS_z$, where
  we have double-counted the $k$-faces of $\partial\pP_1$ and
  $\partial\pP_m$.
  For $k=d$ the set of facets of $\qQ$ is the union of the set of
  facets of $\sS_y$, $\sS_z$ and $\pP$, except for its two facets
  $\pP_1$ and $\pP_m$. We can concisely write the relations described
  above as:
  \begin{equation}\label{equ:fkQ}
    f_k(\qQ)
    =f_k(\pP)+f_k(\sS_y)+f_k(\sS_z)-f_k(\pP_1)-f_k(\pP_m),
    \quad 0\le{}k\le{}d.
  \end{equation}
  where $f_d(\pP_1)=f_d(\pP_m)=1$. The $k$-faces of $\qQ$ in $\sS_z$ are
  either $k$-faces of $\partial\pP_m$ or $k$-faces defined by $z$ and a
  $(k-1)$-face of $\partial\pP_m$. In fact, there exists a bijection between
  the $(k-1)$-faces of $\partial\pP_m$ and the $k$-faces of $\sS_z$ containing
  $z$. Hence, we have:
  \begin{equation}\label{equ:fkSz}
    f_k(\sS_z)=f_k(\partial\pP_m)+f_{k-1}(\partial\pP_m),\quad 0\le{}k\le{}d,
  \end{equation}
  where $f_{-1}(\partial\pP_m)=1$ and $f_d(\partial\pP_m)=0$.
  Analogously, the $k$-faces of $\qQ$ in $\sS_y$ are either $k$-faces
  of $\partial\pP_1$ or $k$-faces defined by $y$ and a $(k-1)$-face of
  $\partial\pP_1$. As for $\sS_z$, there exists a bijection between
  the $(k-1)$-faces of $\partial\pP_1$ and the $k$-faces of $\sS_y$
  containing $y$. Hence, we have:
  \begin{equation}\label{equ:fkSy}
    f_k(\sS_y)=f_k(\partial\pP_1)+f_{k-1}(\partial\pP_1),\quad 0\le{}k\le{}d,
  \end{equation}
  where $f_{-1}(\partial\pP_1)=1$ and $f_d(\partial\pP_1)=0$.
  Therefore, relation \eqref{equ:fkQ} can be rewritten as:
  \begin{equation*}
    f_k(\qQ)=f_k(\pP)+f_k(\partial\pP_1)+f_{k-1}(\partial\pP_1)-f_k(\pP_1)
    +f_k(\partial\pP_m)+f_{k-1}(\partial\pP_m)-f_k(\pP_m),
    \quad 0\le{}k\le{}d.
  \end{equation*}
  or, in the much simpler form:
  \begin{equation}\label{equ:fkQ-simplified}
    f_k(\qQ)=f_k(\pP)+f_{k-1}(\partial\pP_1)+f_{k-1}(\partial\pP_m)-\alpha_k,
    \quad 0\le{}k\le{}d,
  \end{equation}
  where $\alpha_k=0$, for $0\le{}k<d$, whereas $\alpha_d=2$.
  
  Let us now turn our attention on deriving bounds for the face
  numbers $f_k(\qQ)$. Since $\partial\pP_\ell$, $\ell=1,m$, is the
  boundary complex of the $n_\ell$-vertex $d$-polytope $\pP_\ell$, we
  have:
  \begin{equation}\label{equ:fkP1+Pm}
      f_k(\partial\pP_\ell)=O(n_\ell^{\min\{k+1,\lexp{d}\}}),
      \quad 0\le{}k\le{}d,
      \quad\ell=1,m.
  \end{equation}
  Combining the bounds from Lemma \ref{lem:fk-bound-low} with relations
  \eqref{equ:fkQ-simplified} and \eqref{equ:fkP1+Pm}, we get the
  following bounds for the lower-dimensional face numbers of $\qQ$:
  \begin{equation}\label{equ:fboundlowd}
    f_k(\qQ)
    =O(\sum_{1\le{}i\ne{}j\le{}m}n_in_j^k+\sum_{i=1}^{m}n_i^{\min\{k+1,\lexp{d}\}})
    =O(\sum_{1\le{}i\ne{}j\le{}m}n_in_j^{\lexp{d}}),
    \quad 0\le{}k\le{\textstyle\lexp{d}}.
  \end{equation}
  However, $\qQ$ is a simplicial $(d+1)$-polytope, since the facets of
  $\pP$ are simplicial, and since all facets of $\sS_y$ and $\sS_z$
  are simplicial (the facets of $\sS_y$ and $\sS_z$ are defined via
  simplicial $(d-1)$-faces of $\pP$ and the points $y$ and $z$,
  respectively).
  Let us recall the defining equations for the elements of the
  $h$-vector of $\qQ$ in terms of the elements of the $f$-vector
  of $\qQ$:
  \begin{equation}\label{equ:hdef}
    h_k(\qQ)=
    \sum_{i=0}^{k}(-1)^{k-i}\binom{d+1-i}{d+1-k}f_{i-1}(\qQ),\quad
    0\le{}k\le{}d+1.
  \end{equation}
  Combining equations \eqref{equ:hdef} with relations
  \eqref{equ:fboundlowd}, as well as the fact that
  $f_{-1}(\qQ)=1$, we get:
  \begin{equation}\label{equ:hbound}
    h_k(\qQ)=O(\sum_{1\le{}i\ne{}j\le{}m}n_in_j^{\lexp{d}}),
    \quad 0\le{}k\le{}{\textstyle\lexp{d+1}}.
  \end{equation}
  We are now going to use the Dehn-Sommerville equations for
  $\qQ$ to bound the number of $k$-faces $f_k(\qQ)$
  of $\qQ$, for $k\ge\lexp{d+1}$. The Dehn-Sommerville equations
  can be rewritten as follows (cf. \cite[Section 8.4]{z-lp-95}):
  \begin{equation}\label{equ:dsalt}
    f_{k-1}(\qQ)=\sideset{}{^{\,*}}{\sum}_{i=0}^{\frac{d+1}{2}}
    \left(\binom{d+1-i}{k-i}+\binom{i}{k-d-1+i}\right)h_i(\qQ),
    \quad {\textstyle\lexp{d+1}}\le{}k\le{}d+1,
  \end{equation}
  where the symbol
  $\displaystyle\sideset{}{^{\,*}}{\sum}_{i=0}^{\frac{\delta}{2}}$ denotes the
  sum where the last term is halfed if and only if the quantity
  $\frac{\delta}{2}$ is integral (which is our case since in our
  setting $\delta=d+1$, which is even).
  Combining relations \eqref{equ:hbound} and \eqref{equ:dsalt}, we get:
  \begin{equation}\label{equ:fboundhighd}
    f_{k-1}(\qQ)=O(\sum_{1\le{}i\ne{}j\le{}m}n_in_j^{\lexp{d}}),
    \quad {\textstyle\lexp{d+1}}\le{}k\le{}d+1.
  \end{equation}
  Now using relation \eqref{equ:fkQ-simplified}, we arrive at the
  following bounds on the number of $k$-faces of $\pP$ for
  $\lexp{d+1}\le{}k\le{}d$:
  \begin{equation*}
    f_k(\pP)=f_k(\qQ)-f_{k-1}(\partial\pP_1)-f_{k-1}(\partial\pP_m)+\alpha_k
    \le{}f_k(\qQ)
    =O(\sum_{1\le{}i\ne{}j\le{}m}n_in_j^{\lexp{d}}),
  \end{equation*}
  where we used the fact that $f_{d-1}(\partial\pP_1)>\alpha_d$,
  since $f_{d-1}(\partial\pP_1)\ge{}d+1\ge{}4$.
\end{proof}

By Lemmas \ref{lem:fk-bound-low} and \ref{lem:fk-bound-high}, we
deduce that the worst-case complexity of the convex hull
$CH_{d+1}(\PPP)$ is $O(\sum_{1\le{}i\ne{}j\le{}m}n_in_j^{\lexp{d}})$.
As we will see in Subsection \ref{subsec:lbmult} (see Corollary
\ref{cor:polytopeslb}), this bound is asymptotically tight for any odd
$d\ge{}3$. Hence:

\begin{theorem}\label{main}
  Let $\PPP=\{\pP_1,\pP_2,\ldots,\pP_m\}$ be a set of a fixed number
  of $m\ge{}2$ $d$-polytopes, lying on $m$ parallel hyperplanes of
  $\RR^{d+1}$, where $d\ge{}3$ and $d$ is odd. The worst-case
  complexity of the convex hull $CH_{d+1}(\PPP)$ is 
  $\Theta(\sum_{1\le{}i\ne{}j\le{}m}n_in_j^{\lexp{d}})$,
  where $n_i=f_0(\pP_i)$, $1\le{}i\le{}m$.
\end{theorem}


%% file: chp-algo.tex

\section{Computing the convex hull of parallel convex polytopes}
\label{sec:chp-algo}

The upper bound on the worst-case complexity of $CH_{d+1}(\PPP)$ in
Theorem \ref{main} suggests that, in order to compute $CH_{d+1}(\PPP)$
for $d\ge{}3$ odd, it pays off to use an output-sensitive algorithm for
constructing the convex hull of the point set $P$ formed by the
vertices of the $\pP_i$'s.

Let us briefly discuss the output-sensitive algorithms applicable in
our setting; in what follows $n_i$ is the number of vertices of
$\pP_i$, $n$ is the total number of input points and $f$ the size of
the output. The dimension $d+1\ge{}4$ below is the dimension of the
ambient space $\RR^{d+1}$ and is considered to be even; the polytopes
$\pP_i$ are then $d$-dimensional.
One of the earliest algorithms is Seidel's shelling algorithm
\cite{s-chdch-86} that runs in $O(n^2+f\log{}n)$ time.
The preprocessing step of Seidel's algorithm was later on improved by
Matou{\v{s}}ek and Schwarzkopf \cite{ms-loq-92}, resulting in an
$O(n^{2-2/(\lfloor{}(d+1)/2\rfloor+1)+\epsilon}+f\log{}n)$ time algorithm,
for any fixed $\epsilon>0$.
Chan, Snoeyink and Yap \cite{csy-pddpo-97} describe a
divide-and-conquer algorithm for constructing four-dimensional convex
hulls in $O((n+f)\log^2f)$ time.
Finally, Chan \cite{c-osrch-96} improved the gift-wrapping algorithm
of Chand and Kapur \cite{ck-acp-70}, yielding an
$O(n\log{}f+(nf)^{1-1/(\lfloor{}(d+1)/2\rfloor+1)}\log^{O(1)}n)$ time
algorithm.

In our case we have
$f=O(\sum_{1\le{}i\ne{}j\le{}m}n_in_j^{\alpha-1})$, where
$\alpha=\lexp{d+1}\ge{}2$ (recall that we consider only odd $d$).
Applying the algorithms above we get the following. Seidel's algorithm runs in
$O(n^2+(\sum_{1\le{}i\ne{}j\le{}m}n_in_j^{\alpha-1})\log{}n)$ time.
Matou{\v{s}}ek and Schwarzkopf's improvement results in an
$O(n^{2-2/(\alpha+1)+\epsilon}+(\sum_{1\le{}i\ne{}j\le{}m}n_in_j^{\alpha-1})\log{}n)$
time algorithm. Chan's algorithm yields a running time of
$O((\sum_{1\le{}i\ne{}j\le{}m}n_i^{\alpha/(\alpha+1)}n_j^{\alpha^2/(\alpha+1)})
\log^{O(1)}n)$. Finally,
the algorithm by Chan, Snoeyink and Yap gives 
$O((\sum_{1\le{}i\ne{}j\le{}m}n_in_j)\log^2n)$.
For all $d\ge{}3$, \ie for $\alpha\ge{}2$, Chan's algorithm is never
faster than the other two algorithms. For $d\ge{}5$, i.e.,
$\alpha\ge{}3$, it makes no difference to choose between the algorithm
of Seidel and that of Matou{\v{s}}ek and Schwarzkopf.
For $d=3$ the situation is a bit more complicated. For some choices
for the $n_i$'s Matou{\v{s}}ek and Schwarzkopf's algorithm outperforms
all other algorithms, whereas for other choices of the $n_i$'s the
algorithm of Chan, Snoeyink and Yap is the best choice.
Consider, for example, the cases:
\begin{smallenum}
\item[(1)] $n_1=\Theta(n)$,
  $n_2=o(\frac{n^{1/3+\epsilon}}{\log^2{}n})$ and $n_i=O(1)$ for
  $3\le{}i\le{}m$, and
\item[(2)] $n_1=\Theta(n)$,
  $n_2=\omega(\frac{n^{1/3+\epsilon}}{\log{}n})$ and $n_i=O(1)$ for
  $3\le{}i\le{}m$.
\end{smallenum}
In the first case Chan, Snoeyink and Yap's algorithm runs in
$o(n^{4/3+\epsilon})$ time and outperforms all other algorithms. In
the second case, Matou{\v{s}}ek and Schwarzkopf's algorithm seems to
be the best choice, since its running time is $O(n_2n\log{}n)$
(compared to Chan, Snoeyink and Yap's algorithm, the running time of
which is $O(n_2n\log^2{}n)$). 
Summarizing, we arrive at the following:

\begin{theorem}\label{chcp_time}
  Let $\PPP=\{\pP_1,\pP_2,\ldots,\pP_m\}$ be a set of $m$
  $d$-polytopes, lying on $m$ parallel hyperplanes of
  $\RR^{d+1}$, where $d\ge{}3$ and $d$ odd. 
  Let $n_i=f_0(\pP_i)$, $1\le{}i\le{}m$, and
  $n=\sum_{i=1}^mn_i$.
  We can compute the convex hull $CH_{d+1}(\PPP)$ in
  $O((\sum_{1\le{}i\ne{}j\le{}m}n_in_j^{\lexp{d}})\log{}n)$ time for
  $d\ge{}5$, and
  $O(\min\{n^{4/3+\epsilon}+(\sum_{1\le{}i\ne{}j\le{}m}n_in_j)\log{}n,
  (\sum_{1\le{}i\ne{}j\le{}m}n_in_j)\log^2n\})$ time for $d=3$, for
  any fixed $\epsilon>0$.
\end{theorem}

%% file: chs.tex

\section{Convex hulls of spheres with a fixed number of distinct
  radii}
\label{sec:chs}

In this section we derive tight upper and lower bounds on the worst-case
complexity of the convex hull of a set of spheres
in $\RR^d$ having $m$ distinct radii, where $m$ is considered to be fixed.

Let $\Sigma$ be a set of $n$ spheres $\sigma_k=(c_k,r_k)$,
$1\le{}k\le{}n$, in $\RR^d$, and let $CH_d(\Sigma)$ be the convex
hull of the spheres in $\Sigma$. A \emph{face of circularity $\ell$}
of $CH_d(\Sigma)$, $0\le\ell\le{}d-1$, is a maximal connected portion
of the boundary of $CH_d(\Sigma)$ consisting of points where the
supporting hyperplanes are tangent to a given set of $(d-\ell)$
spheres of $\Sigma$.
In the special case where all spheres have the same radius,
$CH_d(\Sigma)$ is combinatorially equivalent to the convex hull
$CH_d(K)$ of the centers $K$ of spheres in $\Sigma$, in the sense that
each face of circularity $\ell$ of $CH_d(\Sigma)$ corresponds to a
unique $(d-\ell-1)$-face of $CH_d(K)$, for $0\le\ell\le{}d-1$.

We consider here the case where the radii $r_k$ can take $m$ distinct
values, i.e., $r_k\in\{\rho_1,\rho_2,\ldots,\rho_m\}$.
Without loss of generality we may assume that
$0<\rho_1<\rho_2<\ldots<\rho_m$.
We identify $\RR^d$ with the hyperplane $H_0=\{x_{d+1}=0\}$ of
$\RR^{d+1}$ and we call the $(d+1)$-axis of $\RR^{d+1}$ the
\emph{vertical axis}, while the expression \emph{above} will refer
to the $(d+1)$-coordinate.
Let $\Pi_i$, $1\le{}i\le{}m$, be the hyperplane $\{x_{d+1}=\rho_i\}$
in $\RR^{d+1}$, and let $P$ be the point set in $\RR^{d+1}$ obtained
by mapping each sphere $\sigma_k$ to the point $p_k=(c_k,r_k)$. Let
$P_i$ denote the subset of $P$ containing points that belong to the
hyperplane $\Pi_i$, and let $n_i$ be the cardinality of $P_i$.
We denote by $\pP$ the convex hull of the points in $P$ (i.e.,
$\pP=CH_{d+1}(P)$). We further denote by $\pP_i$ the convex hull of
the points in $P_i$ (i.e., $\pP_i=CH_d(P_i)$); more precisely, we
identify $\Pi_i$ with $\RR^d$, and then define $\pP_i$ to be the
convex hull of the points in $P_i$, seen as points in $\RR^d$.
We use $\PPP$ to denote the set of the $\pP_i$'s.
Let $\hat{P}_i$ be the subset of $P_i$ that defines $\pP_i$ (i.e., the
points in $\hat{P}_i$ are the vertices of $\pP_i$ and thus
$\pP_i=CH_d(\hat{P}_i)$), and let $\hat{n}_i\le{}n_i$ be the
cardinality of $\hat{P}_i$. Finally, let
$\hat{P}=\bigcup_{i=1}^m\hat{P}_i$ and $\hat{\pP}=CH_{d+1}(\hat{P})$.
Notice that it is possible that $\pP\ne\hat{\pP}$; such a situation
will arise if $P_1\ne{}\hat{P_1}$ (resp., $P_m\ne{}\hat{P_m}$), in
which case the intersection of $\pP$ with $\Pi_1$ (resp., $\Pi_m$)
will consist of more than one $d$-face of $\pP$. On the other hand $\pP$
and $\hat{\pP}$ have the same interior.

Let $\lambda_0$ be the half lower hypercone in $\RR^{d+1}$ with
arbitrary apex, vertical axis, and angle at the apex equal to
$\frac{\pi}{4}$. By $\lambda(p)$ we denote the translated copy of
$\lambda_0$ with apex at $p$; observe that the intersection of
the hypercone $\lambda(p_k)$ with the hyperplane $H_0$ is identical
to the sphere $\sigma_k$. Let $\Lambda$ be the set of the lower
half hypercones $\{\lambda(p_1),\lambda(p_2),\ldots,\lambda(p_n)\}$ in
$\RR^{d+1}$ associated with the spheres of $\Sigma$. The
intersection of the convex hull $CH_{d+1}(\Lambda)$ with $H_0$ is
equal to $CH_d(\Sigma)$.

Let $O'$ be a point in the interior of $\pP$. We then have the following:

\begin{theorem}[{\cite[Theorem 1]{bcddy-acchs-96}}]\label{thm:support}
  Any hyperplane of $\RR^d$ supporting $CH_d(\Sigma)$ is the
  intersection with $H_0$ of a unique hyperplane $H$ of
  $\RR^{d+1}$ satisfying the following three properties:
  \begin{smallenum}
  \item[1.] $H$ supports $\pP$,
  \item[2.] $H$ is the translated copy of a hyperplane tangent to
    $\lambda_0$ along one of its generatrices,
  \item[3.] $H$ is above $O'$. 
  \end{smallenum}
  Conversely, let $H$ be a hyperplane of $\RR^{d+1}$ satisfying the
  above three properties. Its intersection with $H_0$ is a hyperplane
  of $\RR^d$ supporting $CH_d(\Sigma)$.
\end{theorem}

Theorem \ref{thm:support} implies an injection $\varphi:CH_d(\Sigma)\to{}\pP$
that maps each face of circularity $(d-\ell-1)$ of $CH_d(\Sigma)$ to a
unique $\ell$-face of $\pP$, for $0\le\ell\le{}d-1$. 
Theorem \ref{thm:support} also implies that points in
$P_i\sm\hat{P}_i$, $1\le{}i\le{}m$, can never be points on a
supporting hyperplane $H$ of $\pP$ satisfying the three properties of
the theorem. Therefore, $\varphi$ is, in fact, an injection that maps
each face of circularity $(d-\ell-1)$ of $CH_d(\Sigma)$ to a unique
$\ell$-face of $\hat{\pP}$, $0\le\ell\le{}d-1$. Observe that
$\hat{\pP}$ is the convex hull of the set $\PPP$ of $m$ convex
polytopes lying on $m$ parallel hyperplanes of $\RR^{d+1}$. By
employing Theorem \ref{main} of Section \ref{sec:chp},
we deduce that $\hat{\pP}$'s complexity is
$O(\sum_{1\le{}i\ne{}j\le{}m}\hat{n}_i\hat{n}_j^{\lexp{d}})=
O(\sum_{1\le{}i\ne{}j\le{}m}n_in_j^{\lexp{d}})$, which, via the
injection $\varphi:CH_d(\Sigma)\to{}\hat{\pP}$, is also an upper bound
for the worst-case complexity of $CH_d(\Sigma)$.

\subsection{Lower bound construction with two distinct radii}
\label{subsec:lb2}

\newcommand{\tmc}[1]{\gamma_{#1}^{\textsl{\textsf{tr}}}(t)}

For any even dimension $2\delta$, the trigonometric moment curve
$\tmc{2\delta}$ in $\RR^{2\delta}$ is the curve:
\[\tmc{2\delta}=(\cos{}t,\sin{}t,\cos{}2t,\sin{}2t,\ldots,
\cos\delta{}t,\sin\delta{}t),\quad t\in[0,\pi).\]
Notice that points on $\tmc{2\delta}$ are points on the sphere of
$\RR^{2\delta}$ centered at the origin with radius equal to
$\sqrt{\delta}$. For any set $P$ of $n$ points on $\tmc{2\delta}$, the
convex hull $CH_{2\delta}(P)$ is a polytope $\qQ$
combinatorially equivalent to the cyclic polytope $C_{2\delta}(n)$
(cf. \cite{g-cp-03,z-lp-95}). Therefore,
$f_{2\delta-1}(\qQ)=\Theta(n^\delta)$.

\newcommand{\prism}{\Delta}

Suppose now that the ambient space is $\RR^d$, where $d\ge{}3$ is odd.
Let $H_1$ be $H_2$ be the hyperplanes $\{x_d=z_1\}$ and $\{x_d=z_2\}$,
where $z_1,z_2\in\mathbb{R}$ and $z_2>z_1+2(n_2+2)\sqrt{\delta}$; the
quantity $n_2$ will be defined below.
Consider a set $\Sigma_1$ of $n_1+1$ points, treated as spheres of
$\RR^d$ of zero radius, on the $(d-1)$-dimensional trigonometric
moment curve $\tmc{d-1}$ embedded in $H_1$ (please refer to
Fig. \ref{fig:lb}(left), as well as Fig. \ref{fig:lbtop1} for the view
of the construction from the positive $x_d$-axis).
Among the $n_1+1$ points, the first $n_1$ points are chosen with
$t\in(0,\frac{\pi}{2})$, whereas for the remaining point we require
that $t\in(\frac{\pi}{2},\pi)$. This implies that the
$x_1$-coordinate of the first $n_1$ points of $\Sigma_1$ is 
positive, whereas the $x_1$-coordinate of the last point of $\Sigma_1$
is negative. Let $\Sigma_2$ be the projection, along the
$x_d$-axis, of $\Sigma_1$ on the hyperplane $H_2$.
Clearly, the $n_1+1$ points of $\Sigma_2$ in $H_2$ lie on the
$(d-1)$-dimensional trigonometric moment curve $\tmc{d-1}$ embedded in
$H_2$.
The points of $\Sigma_i$, $i=1,2$, lie on a $(d-2)$-dimensional
sphere of $\RR^d$, centered at the point $(0,0,\ldots,0,z_i)$,
with radius $\sqrt{\delta}$. Moreover, the number of facets of the
polytope $\qQ_i=CH_{d-1}(\Sigma_i)$ is
$\Theta(n_1^{\lexp{d-1}})=\Theta(n_1^{\lexp{d}})$.
The convex hull of the $2(n_1+1)$ points of
$\Sigma_1\cup{}\Sigma_2$ is a prism $\prism$. $\prism$ consists of
$\Theta(n_1^{\lexp{d}})$ facets not lying on $H_1$ or $H_2$,
called the \emph{vertical facets} of the prism, the $(d-1)$-face of
$CH_{d-1}(\Sigma_1)$, called the \emph{bottom facet}, and the
$(d-1)$-face of $CH_{d-1}(\Sigma_2)$, called the \emph{top facet}. 
For each vertical facet $F$ of $\prism$, we denote by $\vec{\nu}_F$
the unit normal vector of $F$ pointing outside $\prism$, and by $F^+$
(\resp $F^-$) the positive (\resp negative) open halfspace delimited
by the supporting hyperplane of $F$. Regarding the ridges of $\prism$,
those that are intersections of vertical facets of $\prism$ will be
referred to as \emph{vertical ridges}. Notice that the vertical ridges
of $\prism$ are perpendicular to $H_1$ and $H_2$.

Let $Y$ be the oriented hyperplane $\{x_1=0\}$ with unit normal
vector $\vec{\nu}=(1,0,\ldots,0)$.
Let also $Y^+$ and $Y^-$ be the positive and negative open
halfspaces of $\RR^d$ delimited by $Y$, respectively.
$Y$ contains the $x_d$-axis, and is perpendicular to the hyperplanes
$H_1$ and $H_2$. Recall that $n_1$ points of $\Sigma_i$ are
contained in $Y^+$, whereas exactly one point of $\Sigma_i$ is
contained in $Y^-$. Clearly, $Y$ is in general position with respect
to $\prism$, $\qQ_1$ and $\qQ_2$. Let $\tilde{\qQ}_i$ be the intersection
of $\qQ_i$ with $Y$, and let $\fF_i$ be the set of faces of $\qQ_i$
intersected by $Y$. $\tilde{\qQ}_i$ is a $(d-2)$-polytope, and its
number of vertices is at most $n_1$, since $Y$ cuts at most $n_1$
edges of $\qQ_i$. This implies that the complexity of
$\tilde{\qQ}_i$ is $O(n_1^{\lexp{d-2}})=O(n_1^{\lexp{d}-1})$; the same
bound holds for $\fF_i$. Since there are
no facets of $\prism$ in $Y^-$ ($Y^-$ contains a single vertex of
$\qQ_i$), and since the number of facets of $\qQ_i$ is
$\Theta(n_1^{\lexp{d}})$, we conclude that the number of facets of
$\qQ_i$ contained in $Y^+$ is also $\Theta(n_1^{\lexp{d}})$; the same
bound holds for the number of vertical facets of $\prism$ in $Y^+$.

Define now a set
$\Sigma_3=\{\sigma_0,\sigma_1,\ldots,\sigma_{n_2+1}\}$ of $n_2+2$
spheres in $\RR^d$, where $\sigma_k=(c_k,\rho)$, and
$c_k=(0,\ldots,0,(2k+1)\sqrt{\delta})$,
$0\le{}k\le{}n_2+1$. In other words, the sphere $\sigma_k$ is
centered on the $x_d$-axis, with the $d$-th coordinate of its center $c_k$
being $(2k+1)\sqrt{\delta}$, while its radius is $\rho$. 
We choose $\rho$ to be smaller than $\sqrt{\delta}$, but large enough
so that each sphere $\sigma_i$ satisfies the following two properties:
\begin{smallenum}
\item[(1)] it does not intersect any of the ridges of $\prism$
  (including the vertical ridges of $\prism$), and
\item[(2)] it intersects the interior of all vertical facets of
  $\prism$.
\end{smallenum}
Notice also that for this choice for $\rho$, none of the spheres
in $\Sigma_3$ intersects the hyperplanes $H_1$ and $H_2$ (recall that
$z_2>z_1+2(n_2+2)\sqrt{\delta}$), while the spheres in $\Sigma_3$
are pairwise disjoint; these two observations, however, are not
critical for our construction.

We are now going to perturb the centers of the spheres in
$\Sigma_3$ to get a new set of spheres $\Sigma_3'$ (see
Fig. \ref{fig:lb}(right), as well as Fig. \ref{fig:lbtop2} for the
view of the construction from the positive $x_d$-axis).
Define $\sigma_k'$ to be the sphere with
radius $\rho$ and center
$c_k'=c_k+(\sum_{\ell=0}^k\frac{\varepsilon}{2^\ell})\vec{\nu} 
=c_k+\varepsilon(2-\frac{1}{2^k})\vec{\nu}$, where
$0<\varepsilon\ll{}1$.
The quantity $\varepsilon$ is chosen so that the spheres in
$\Sigma_3'$ satisfy almost the same conditions as the spheres
in $\Sigma_3$.
In particular, we require that condition (1) is still satisfied, while
we relax the requirement on condition (2): we now require that
$\sigma_k'$ intersects the interior of all vertical facets of $\prism$
contained in $Y^+$. In addition to the two conditions
above, we also require that for each $k$, $0\le{}k\le{}n_2+1$, the
$(d-2)$-dimensional sphere $\sigma_k\cap\sigma_k'$ is contained in
$F^-$ for all vertical facets $F$ of $\prism$ in $Y^+$.

\newcommand{\labelsize}{\footnotesize}

\begin{figure}[pt]
  \begin{center}
    \psfrag{Y}[][]{\textcolor{MyGreen}{\labelsize$Y$}}
    \psfrag{x1}[][]{\labelsize$x_1$}
    \psfrag{xd}[][]{\labelsize$x_d$}
    \psfrag{nu}[][]{\textcolor{MyGreen}{\labelsize$\vec{\nu}$}}
    \psfrag{nuF}[][]{\textcolor{MyBlue}{\labelsize$\vec{\nu}_F$}}
    \psfrag{F}[][]{\textcolor{MyBlue}{\labelsize$F$}}
    \includegraphics[width=0.9\textwidth]{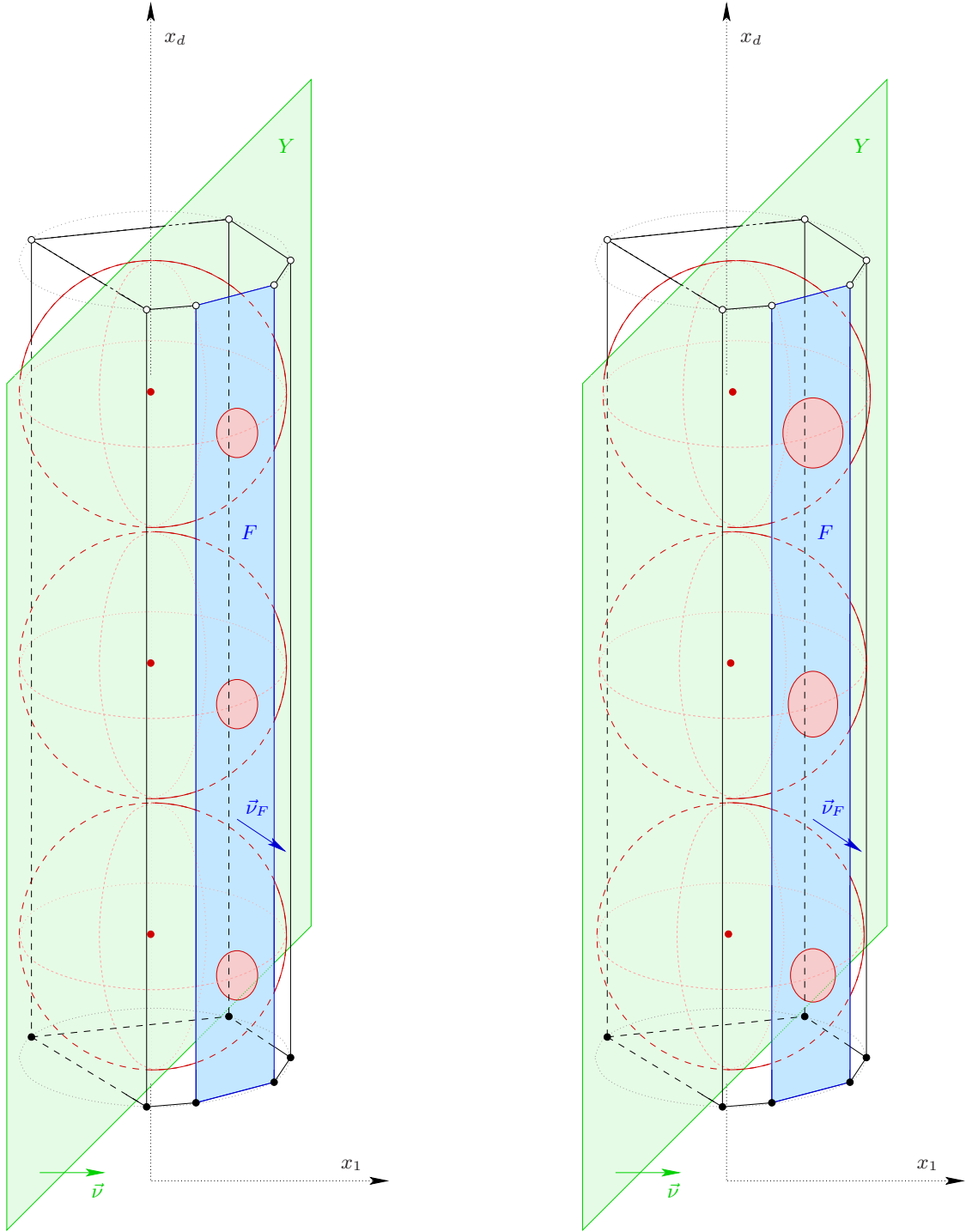}
  \end{center}
  \caption{The lower bound construction in the case of two radii. The
    points in $\Sigma_1$ (\resp $\Sigma_2$) are shown in black (\resp
    white). The hyperplane $Y$ is shown in green, while the prism
    $\prism$ is shown in black. The facet $F$ in blue is one of
    the vertical facets of $\prism$ in $Y^+$. The sphere sets
    $\Sigma_3$ (left) and $\Sigma_3'$ (right) are shown in red.
    The red spherical caps on the left correspond to a unique
    supporting hyperplane of
    $CH_d(\Sigma_1\cup\Sigma_2\cup\Sigma_3)$.
    The red spherical caps on the right correspond to faces of
    $CH_d(\Sigma_1\cup\Sigma_2\cup\Sigma_3')$ of circularity $(d-1)$.}
  \label{fig:lb}
\end{figure}

\renewcommand{\labelsize}{}

\begin{figure}[pt]
  \begin{center}
    \psfrag{Y}[][]{\textcolor{MyGreen}{\labelsize$Y$}}
    \psfrag{x1}[][]{\labelsize$x_1$}
    \psfrag{xd}[][]{\labelsize$x_d$}
    \psfrag{nu}[][]{\textcolor{MyGreen}{\labelsize$\vec{\nu}$}}
    \psfrag{nuF}[][]{\textcolor{MyBlue}{\labelsize$\vec{\nu}_F$}}
    \psfrag{F}[][]{\textcolor{MyBlue}{\labelsize$F$}}
    \includegraphics[width=\textwidth]{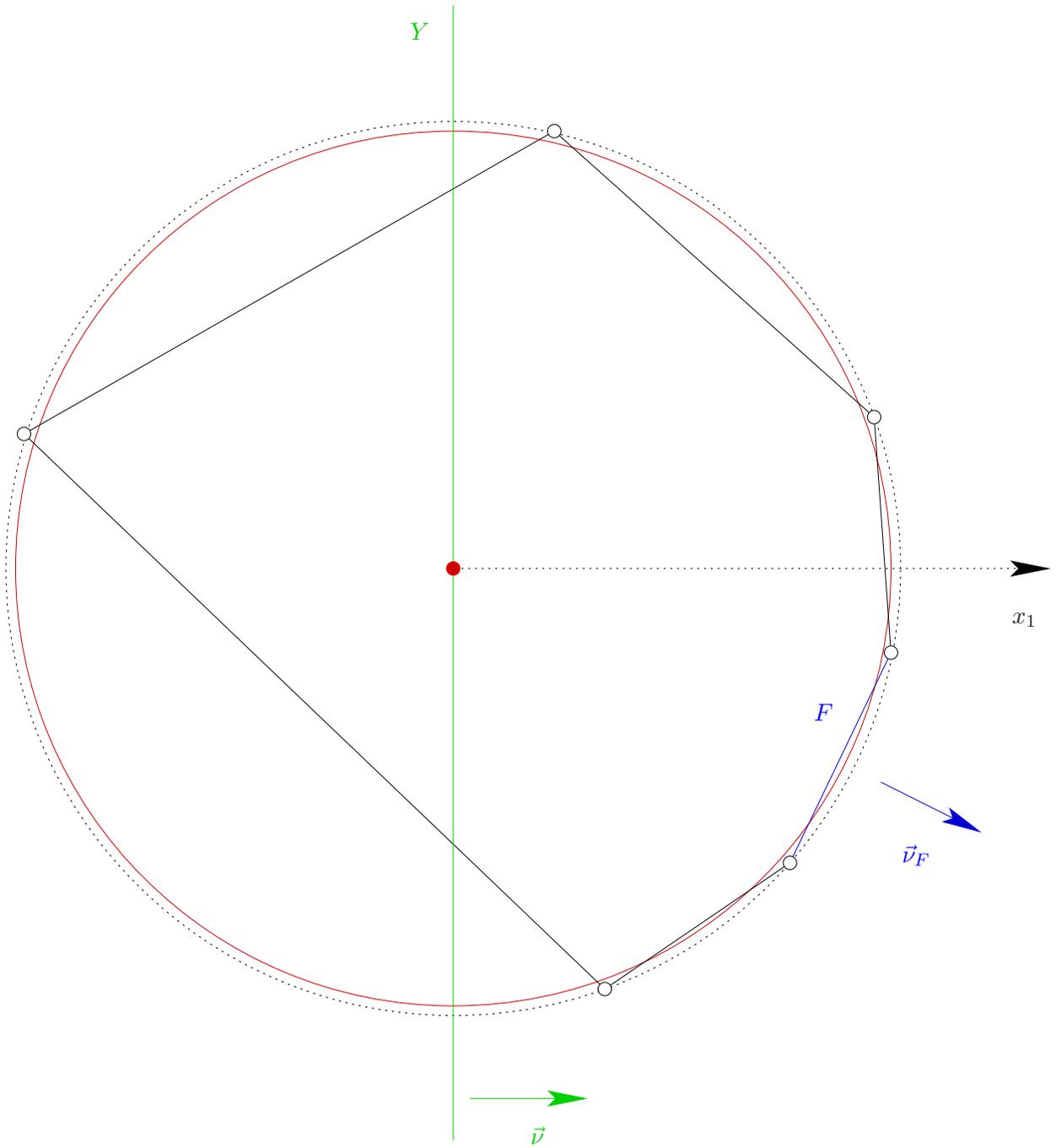}
  \end{center}
  \caption{View from the positive $x_d$-axis of the construction in 
    Fig. \ref{fig:lb}(left). The silhouettes of all spheres in
    $\Sigma_3$ coincide.}
  \label{fig:lbtop1}
\end{figure}

\begin{figure}[pt]
  \begin{center}
    \psfrag{Y}[][]{\textcolor{MyGreen}{\labelsize$Y$}}
    \psfrag{x1}[][]{\labelsize$x_1$}
    \psfrag{xd}[][]{\labelsize$x_d$}
    \psfrag{nu}[][]{\textcolor{MyGreen}{\labelsize$\vec{\nu}$}}
    \psfrag{nuF}[][]{\textcolor{MyBlue}{\labelsize$\vec{\nu}_F$}}
    \psfrag{F}[][]{\textcolor{MyBlue}{\labelsize$F$}}
    \includegraphics[width=\textwidth]{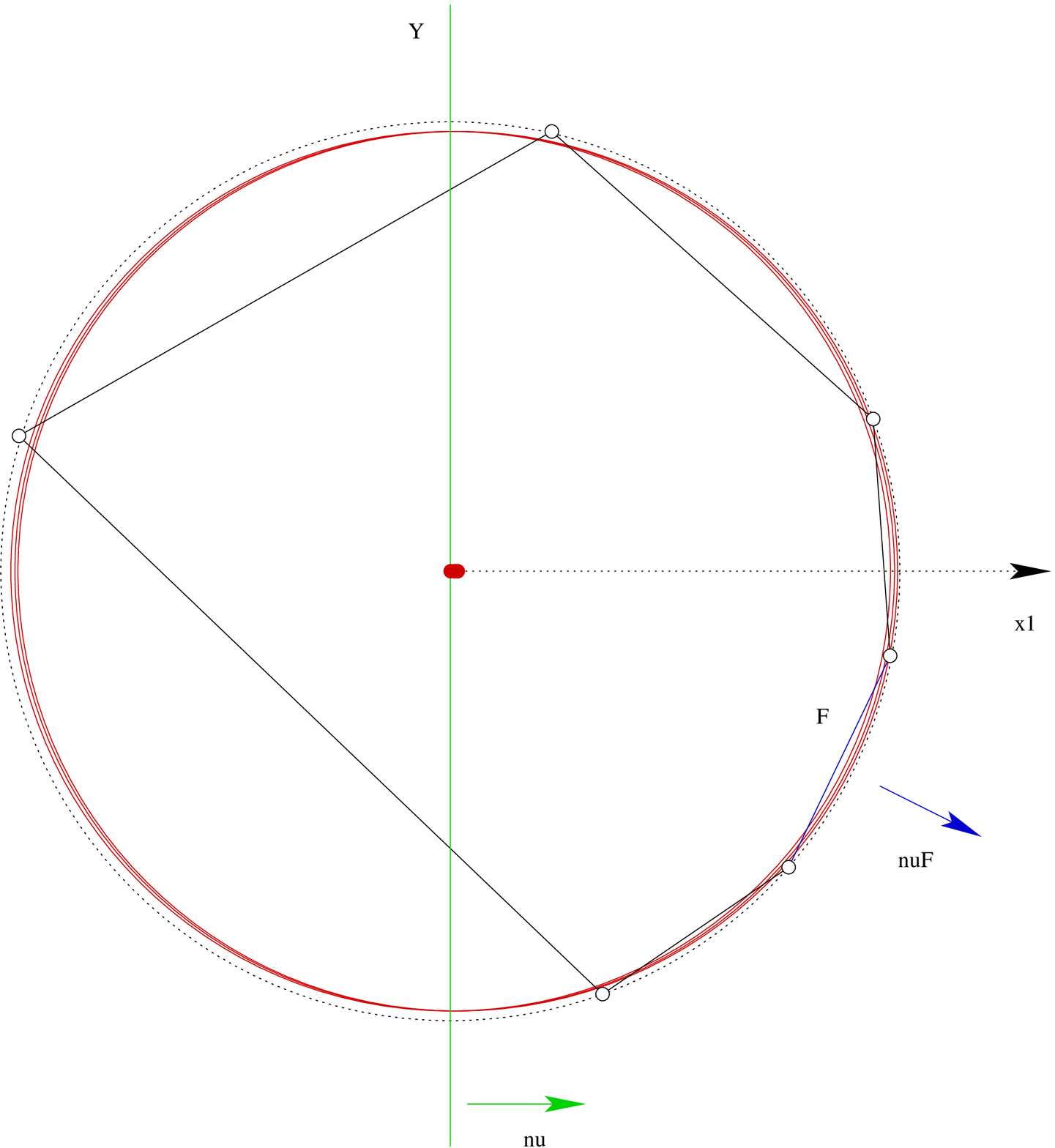}
  \end{center}
  \caption{View from the positive $x_d$-axis of the construction in 
    Fig. \ref{fig:lb}(right).}
  \label{fig:lbtop2}
\end{figure}

We will now show that for each pair $(\sigma_k',F)$, where
$1\le{}k\le{}n_2$ and $F$ is a vertical facet of $\prism$ in $Y^+$,
the spherical cap $F^+\cap{}\sigma_k'$ induces a face of circularity
$(d-1)$ in $CH_d(\Sigma)$.
Let $F_1$ and $F_2$ be the ridges of $\prism$
on the boundary of $F$ contained in the top and bottom facet,
respectively. 
Finally, let $S_k$ be the supporting hyperplane of $\sigma_k$ parallel
to $F$; we consider $S_k$ to be oriented as $F$ (\ie the unit
normal vector of $S_k$ is $\vec{\nu}_F$), and thus $\sigma_k$ lies in
the closure of the negative halfspace delimited by $S_k$. Notice that
$S_k$ is also a supporting hyperplane for $\Sigma_3$. Let $S_k'$ be
the hyperplane we get by translating $S_k$ by the vector
$\varepsilon(2-\frac{1}{2^k})\vec{\nu}$. $S_k'$ supports $\sigma_k'$,
but fails to be a supporting hyperplane for $\Sigma_3'$. More
precisely, $S_k'$ intersects all spheres $\sigma_j'$ with $j>k$,
whereas all spheres $\sigma_j'$ with $j<k$, are contained in the
negative open halfspace delimited by $S_k'$.
We can, however, perturb $S_k'$ so that it supports $\Sigma_3'$:
simply slide $S_k'$ on sphere $\sigma_k'$ towards $F_1$,
while maintaining the property that it remains parallel to $F_1$ and
$F_2$. We keep sliding $S_k'$ until it has empty intersection with any
sphere $\sigma_j'$ with $j>k$. Notice that due to the way we have
perturbed the centers of the spheres in $\Sigma_3$ to get $\Sigma_3'$,
the new hyperplane $S_k''$ we get via this transformation is a
supporting hyperplane for $\Sigma_3'$. In fact, $S_k''$
is a supporting hyperplane for the sphere set
$\Sigma=\Sigma_1\cup\Sigma_2\cup\Sigma_3'$ (it touches
$CH_d(\Sigma)$ at $\sigma_k'$ only), which implies that $S_k''$
corresponds to a unique face of circularity $(d-1)$ on
$CH_d(\Sigma)$.

The same construction can be done for all $k$ with $1\le{}k\le{}n_2$,
and for all vertical facets of $\prism$ in $Y^+$. Since we have
$\Theta(n_1^{\lexp{d}})$ vertical facets of $\prism$ in $Y^+$, we
can construct $n_2\Theta(n_1^{\lexp{d}})$ distinct supporting hyperplanes
of $CH_d(\Sigma)$, corresponding to distinct faces of circularity
$(d-1)$ on $CH_d(\Sigma)$. Hence the complexity of
$CH_d(\Sigma)$ is $\Omega(n_2n_1^{\lexp{d}})$.
Without loss of generality, we may assume that $n_2\le{}n_1$, in which
case we have
$n_2n_1^{\lexp{d}}\ge{}\frac{1}{2}(n_2n_1^{\lexp{d}}+n_1n_2^{\lexp{d}})$.
Hence, we arrive at the following:

\begin{theorem}\label{thm:chs-lb2}
  Fix some odd $d\ge{}3$. There exists a set $\Sigma$ of
  spheres in $\RR^d$, consisting of $n_i$ spheres of
  radius $\rho_i$, $i=1,2$, with $\rho_1\ne\rho_2$, such that
  the complexity of the convex hull $CH_d(\Sigma)$ is
  $\Omega(n_1n_2^{\lexp{d}}+n_2n_1^{\lexp{d}})$.
\end{theorem}

\subsection{Lower bound construction with \texorpdfstring{$m$}{m}
  distinct radii}
\label{subsec:lbmult}

We can easily generalize the lower bound construction of the previous
subsection in the case where we have $n_i$ spheres of radius
$\rho_i$, $1\le{}i\le{}m$, $m\ge{}3$, and the radii $\rho_i$ are
considered to be mutually distinct.

As in the previous subsection, the ambient space is $\RR^d$, where
$d\ge{}3$ is odd. Let $N_1=\sum_{i=2}^mn_i$ and $N_2=n_1$. We construct
the set $\Sigma=\Sigma_1\cup\Sigma_2\cup\Sigma_3'$ as in the previous
subsection where $\Sigma_1$ and $\Sigma_2$ contain each $N_1+1$ points
and $\Sigma_3'$ contains $N_2+2$ spheres of some appropriate
radius $\rho$ (recall that in the construction of the previous
subsection $\rho\approx\sqrt{\frac{d-1}{2}}\ge{}1$).
We then replace $n_i$ among the $N_1$ points of $\Sigma_1$ (\resp
$\Sigma_2$) contained in $Y^+$ by spheres with the same center
and radius equal to $r^i$, where $0<r\ll{}1$. Furthermore, we
replace the unique point of $\Sigma_1$ (\resp $\Sigma_2$) in $Y^-$ by
a sphere of the same center and radius $r^2$.
We choose $r$ small enough so that the following two conditions
hold:
\begin{smallenum}
\item[(1)] the prism $\prism_r=CH_d(\Sigma_1\cup\Sigma_2)$ is
  combinatorially equivalent\footnote{Combinatorial equivalence here
    means that each face of circularity $\ell$ of $\prism_r$
    corresponds to a unique $(d-\ell-1)$-face of $\prism_0$.}
  to the prism $\prism_0$ (this is the
  prism we get for $r=0$, which is the prism $\prism$ of the
  previous subsection), and
\item[(2)] the two requirements for the spheres in $\Sigma_3'$
  are still satisfied: each $\sigma_k'$ does not intersect
  any of the ridges of $\prism_r$, while each $\sigma_k'$ intersects
  the interior of all vertical facets of $\prism_r$ contained in
  $Y^+$.
\end{smallenum}
As described in the previous subsection, the convex hull
$CH_d(\Sigma)$ of the set $\Sigma=\Sigma_1\cup\Sigma_2\cup\Sigma_3'$
of $2(N_1+1)+N_2+2$ spheres has $N_2\Theta(N_1^{\lexp{d}})$ faces
of circularity $(d-1)$, and hence its complexity is
$\Omega(N_2N_1^{\lexp{d}})=\Omega(n_1(\sum_{i=2}^mn_i)^{\lexp{d}})$.
Without loss of generality we may assume that
$n_2\ge{}n_1\ge{}n_i$ for all $3\le{}i\le{}m$, in which case we have:
$n_1(\sum_{i=2}^mn_i)^{\lexp{d}}\ge{}n_1n_2^{\lexp{d}}\ge{}
\frac{1}{m(m-1)}(\sum_{1\le{}i\ne{}j\le{}m}n_in_j^{\lexp{d}})$.
Since $m$ is fixed, we conclude that the complexity of
$CH_d(\Sigma)$ is $\Omega(\sum_{1\le{}i\ne{}j\le{}m}n_in_j^{\lexp{d}})$.

\begin{theorem}\label{thm:chs-lbm}
  Fix some odd $d\ge{}3$. There exists a set $\Sigma$ of
  spheres in $\RR^d$, consisting of $n_i$ spheres of radius
  $\rho_i$, $1\le{}i\le{}m$, with $\rho_1<\rho_2<\ldots<\rho_m$ and
  $m\ge{}3$ fixed, such that the complexity of the convex hull
  $CH_d(\Sigma)$ is
  $\Omega(\sum_{1\le{}i\ne{}j\le{}m}n_in_j^{\lexp{d}})$. 
\end{theorem}

Consider again the injection
$\varphi:CH_d(\Sigma)\rightarrow\hat{\pP}$. We have shown above
that the worst-case complexity of $CH_d(\Sigma)$ is
$\Omega(\sum_{1\le{}i\ne{}j\le{}m}n_in_j^{\lexp{d}})$, when $d\ge{}3$ is
odd and $m\ge{}2$ is fixed. Since $\varphi$ is injective, this lower
bound also applies to the complexity of
$\hat{\pP}$. This establishes our lower bound claim in
Theorem \ref{main}:

\begin{corollary}\label{cor:polytopeslb}
  Let $\PPP=\{\pP_1,\pP_2,\ldots,\pP_m\}$ be a set of $m$
  $d$-polytopes, lying on $m$ parallel hyperplanes of $\RR^{d+1}$,
  with $d\ge{}3$, $d$ odd, and both $d$ and $m$ are fixed.
  The worst-case complexity of $CH_{d+1}(\PPP)$ is
  $\Omega(\sum_{1\le{}i\ne{}j\le{}m}n_in_j^{\lexp{d}})$,
  where $n_i=f_0(\pP_i)$, $1\le{}i\le{}m$.
\end{corollary}

%% file: chs-algo.tex

\section{Computing convex hulls of spheres}
\label{sec:chs-algo}

In this section we focus our attention to the computation of the
convex hull $CH_d(\Sigma)$ of $\Sigma$. We use the same notation as in
Section \ref{sec:chs}.

Polarity with respect to the point $O'$ is a one-to-one
transformation which maps points of $\RR^{d+1}$ distinct from $O'$ to
hyperplanes of $\RR^{d+1}$ that do not contain $O'$ (e.g., see
\cite[Section 7.1.3]{by-ag-98}). For a point $q$
in $\RR^{d+1}$ distinct from $O'$, we define its
\emph{polar hyperplane} $\polar{q}$ by
{
  \setlength\abovedisplayskip{4pt plus 3pt minus 7pt}
  \setlength\belowdisplayskip{4pt plus 3pt minus 7pt}
  \[\polar{q}=\{X\in\RR^{d+1}\,|\,q\cdot{}X=1\},\]
}%
while for a hyperplane $H$ of $\RR^{d+1}$ not containing $O'$,
the \emph{polar point} or \emph{pole} $\polar{H}$ of $H$ is defined by
{
  \setlength\abovedisplayskip{4pt plus 3pt minus 7pt}
  \setlength\belowdisplayskip{4pt plus 3pt minus 7pt}
  \[\polar{H}\cdot{}X=1,\quad \forall X\in{}H.\]
}%
The \emph{polar set} of a set of hyperplanes is the set
of the poles of these hyperplanes. Finally, we denote by $H^-$ the
closed halfspace bounded by $H$ and containing $O'$. 

Consider again the point set $P$ in $\RR^{d+1}$ that we get by
lifting the spheres of $\Sigma$. The polytope
$\polar{\pP}$ in $\RR^{d+1}$ defined as the intersection of the
halfspaces $\polarhp{p_k}$, $1\le{}k\le{}n$, is the
\emph{polar polytope} or \emph{dual polytope} of $\pP$. 
$\pP$ and $\polar{\pP}$ are dual in the sense that there is a
bijection between the $\ell$-faces of $\pP$ and the $(d-\ell)$-faces
of $\polar{\pP}$, which is inclusion-reversing: a hyperplane
supporting $\pP$ along an $\ell$-face $F$, has its pole on the
corresponding $(d-\ell)$-face $\polar{F}$ of $\polar{\pP}$. 
In a similar manner, we can define the polytope $\polar{\hat{\pP}}$,
which is the polar polytope of $\hat{\pP}$.

Let $H_{O'}$ be the hyperplane parallel to $H_0$ passing through $O'$
and let $H_{O'}^+$ denote the open halfspace of $\RR^{d+1}$ above
$H_{O'}$. Boissonnat \etal \cite[Proposition]{bcddy-acchs-96} have
shown that:\medskip
\begin{smallenum}
\item[1.] The polar set of the hyperplanes, which are translated copies
  of the hyperplanes tangent to $\lambda_0$, is the hypercone $K$ with
  apex at $O'$, vertical axis, and angle at the apex equal to
  $\frac{\pi}{4}$.
\item[2.] The polar set of the hyperplanes above $O'$ is the halfspace
  $H_{O'}^+$.
\end{smallenum}\medskip
An immediate consequence of these properties is that the polar set of
the hyperplanes supporting $\pP$, tangent to at least one hypercone of
$\Lambda$ along a generatrix, and above $O'$, is
$\polar{\pP}\cap{}K\cap{}H_{O'}^+$.

Boissonnat \etal \cite{bcddy-acchs-96} have used this property in
order to propose an algorithm for computing $CH_d(\Sigma)$ in
$O(n^{\cexp{d}}+n\log{}n)$ time for any $d\ge{}2$.
Below, we describe a slightly modified
algorithm that takes into account the fact that the radii of the
spheres in $\Sigma$ can take a fixed number of $m\ge{}2$ distinct
values. Our algorithm consists of the following five steps:\medskip
\begin{smallenum}
\item[1.] For all $i$ with $1\le{}i\le{}m$: determine the set
  $P_i=P\cap\Pi_i$, construct the convex hull $\pP_i=CH_d(P_i)$, and
  compute the set $\hat{P}_i$.
\item[2.] Compute the polytope
  $\hat{\pP}=CH_{d+1}(\hat{P})$, and choose a point $O'$
  inside $\hat{\pP}$.
\item[3.] Compute the polar polytope $\polar{\hat{\pP}}$ of
  $\hat{\pP}$ with respect to $O'$.
\item[4.] Compute the intersection between $\polar{\hat{\pP}}$, the
  hypercone $K$, and the halfspace $H_{O'}^+$.
\item[5.] Compute the incidence graph of the facets in $CH_d(\Sigma)$
  from the incidence graph of the faces of $\polar{\hat{\pP}}$
  intersecting $K$ and $H_{O'}^+$.
\end{smallenum}\medskip

Determining all the sets $P_i$ takes $\Theta(n)$ time,
whereas constructing the polytope $\pP_i$ takes
$O(n_i^{\lexp{d}}+n_i\log{}n_i)$ time. Determining $\hat{P}_i$ from
from the representation of $\pP_i$ can easily be done in
$O(n_i\log{}n_i)$ time. We thus conclude that step 1 of the algorithm
takes $O(n^{\lexp{d}}+n\log{}n)$ time.
Computing $\polar{\hat{\pP}}$ from $\hat{\pP}$ takes linear time in
the size of $\hat{\pP}$. Moreover, the complexity of
computing $\polar{\hat{\pP}}\cap{}K\cap{}H_{O'}^+$ is linear in the
complexity of $\polar{\hat{\pP}}$, and thus $\hat{\pP}$, since $K$ and
$H_{O'}^+$ have constant complexity. The complexity for computing the
incidence graph of $CH_d(\Sigma)$ from the incidence graph of
$\polar{\hat{\pP}}\cap{}K\cap{}H_{O'}^+$ is again linear in
$\polar{\hat{\pP}}$ and $\hat{\pP}$. Therefore, the total time needed
for steps 2--5 of the algorithm is a linear function of the time
needed to compute $\hat{\pP}$. Summarizing:

\begin{theorem}\label{thm:chs-time}
  Let $\Sigma$ be a set of $n$ spheres in $\RR^d$, having $m$
  distinct radii $\rho_1,\rho_2,\ldots,\rho_m$, with $d\ge{}3$, $d$
  odd, and both $d$ and $m$ are fixed. Let $n_i$ be the number of
  spheres in $\Sigma$ with radius $\rho_i$, $1\le{}i\le{}m$.
  We can compute the convex hull $CH_d(\Sigma)$ in
  $O(n^{\lexp{d}}+n\log{}n+T_{d+1}(n_1,n_2\ldots,n_m))$ time, where
  $T_{d+1}(n_1,n_2\ldots,n_m)$ stands for the time to compute the
  convex hull of $m$ $d$-polytopes $\{\pP_1,\pP_2,\ldots,\pP_m\}$
  lying on $m$ parallel hyperplanes in $\RR^{d+1}$, with
  $n_i=f_0(\pP_i)$, $1\le{}i\le{}m$.
\end{theorem}

As described in Section \ref{sec:chp-algo},
\[
T_{d+1}(n_1,\ldots,n_m)=
O(({\textstyle\sum_{1\le{}i\ne{}j\le{}m}}n_in_j^{\lexp{d}})\log{}n),\]
for any odd $d\ge{}5$, whereas for $d=3$ we have
\[T_4(n_1,\ldots,n_m)=O(\min\{n^{4/3+\epsilon}+
({\textstyle\sum_{1\le{}i\ne{}j\le{}m}}n_in_j^{\lexp{d}})\log{}n,
({\textstyle\sum_{1\le{}i\ne{}j\le{}m}}n_in_j^{\lexp{d}})\log^2n\}).\]
Hence, we can compute the convex hull $CH_d(\Sigma)$ in
$O((\sum_{1\le{}i\ne{}j\le{}m}n_in_j^{\lexp{d}})\log{}n)$ time for
any odd $d\ge{}5$, and in
$O(\min\{n^{4/3+\epsilon}+(\sum_{1\le{}i\ne{}j\le{}m}n_in_j^{\lexp{d}})\log{}n,
(\sum_{1\le{}i\ne{}j\le{}m}n_in_j^{\lexp{d}})\log^2n\})$ time for $d=3$.

%% file: concl.tex

\section{Summary and open problems}
\label{sec:concl}

In this paper we have considered the problem of computing the
worst-case complexity of the convex hull $CH_{d+1}(\PPP)$ of a set
$\PPP=\{\pP_1,\pP_2,\ldots,\pP_m\}$ of $m$ convex $d$-polytopes lying
on $m$ parallel hyperplanes of $\RR^{d+1}$, for any odd
$d\ge{}3$. Denoting by $n_i$ the number of vertices of $\pP_i$, we
have shown that the worst-case complexity of $CH_{d+1}(\PPP)$ is
$O(\sum_{1\le{}j\ne{}j\le{}m}n_in_j^{\lexp{d}})$. This result
suggests that, in order to compute $CH_{d+1}(\PPP)$, it pays off to
apply an output-sensitive convex hull algorithm to the set of vertices
in $\PPP$. Indeed, for any odd $d\ge{}5$ Seidel's shelling algorithm
\cite{s-chdch-86}, or its modification by Matou{\v{s}}ek and
Schwarzkopf \cite{ms-loq-92}, results in a
$O((\sum_{1\le{}j\ne{}j\le{}m}n_in_j^{\lexp{d}})\log{}n)$ time
algorithm, where $n=\sum_{i=1}^mn_i$. For $d=3$, the
divide-and-conquer algorithm by  Chan, Snoeyink and Yap
\cite{csy-pddpo-97} can be competitive against
Matou{\v{s}}ek and Schwarzkopf's algorithm; hence, we may compute
$CH_4(\PPP)$ in
$O(\min\{n^{4/3+\epsilon}+(\sum_{1\le{}i\ne{}j\le{}m}n_in_j)\log{}n,
(\sum_{1\le{}i\ne{}j\le{}m}n_in_j)\log^2n\})$ time, for
any fixed $\epsilon>0$. The above algorithms are nearly optimal for
any odd $d\ge{}3$; it remains an open problem to
compute $CH_{d+1}(\PPP)$ in worst-case optimal
$O(\sum_{1\le{}j\ne{}j\le{}m}n_in_j^{\lexp{d}}+n\log{}n)$ time.

A direct consequence of our bound on the complexity of
$CH_{d+1}(\PPP)$ is a tight asymptotic bound on the worst-case
complexity of the (weighted) Minkowski sum of two polytopes in any odd
dimension $d\ge{}3$.
More precisely, consider a $n$-vertex $d$-polytope $\pP$ and a
$m$-vertex $d$-polytope $\qQ$, and embed them on the hyperplanes
$\{x_{d+1}=0\}$ and $\{x_{d+1}=1\}$ of $\RR^{d+1}$, respectively.
The weighted Minkowski sum $(1-\lambda)\pP\oplus\lambda\qQ$,
$\lambda\in(0,1)$, is combinatorially equivalent to the intersection
of $CH_{d+1}(\{\pP,\qQ\})$ with the hyperplane $\{x_{d+1}=\lambda\}$,
whereas the Minkowski sum $\pP\oplus\qQ$ is nothing but
$\frac{1}{2}\pP\oplus\frac{1}{2}\qQ$, scaled by a factor of
2. Applying our results, we deduce that the complexity of
$(1-\lambda)\pP\oplus\lambda\qQ$ (\resp $\pP\oplus\qQ$) is
$\Theta(nm^{\lexp{d}}+mn^{\lexp{d}})$ for any odd $d\ge{}3$. 
We would like to extend this tight bound to (weighted) Minkowski sums
where the number of summands is greater than 2.

Capitalizing on our result on the complexity of convex hulls of convex
polytopes lying on parallel hyperplanes, we have shown that the
worst-case complexity of the convex hull $CH_d(\Sigma)$ of a set
$\Sigma$ of $n$ spheres in $\RR^d$ with a fixed number of $m$
distinct radii $\rho_1,\rho_2,\ldots,\rho_m$ is
$O(\sum_{1\le{}j\ne{}j\le{}m}n_in_j^{\lexp{d}})$, for 
any odd $d\ge{}3$, where $n_i$ is the number of spheres with
radius $\rho_i$.
By means of an appropriate construction, described in
Subsections \ref{subsec:lb2} and \ref{subsec:lbmult}, we have
shown that the upper bound above is asymptotically tight, implying
that our upper bound for $CH_{d+1}(\PPP)$ is also tight. By slightly,
but crucially, modifying the algorithm of Boissonnat \etal
\cite{bcddy-acchs-96}, $CH_d(\Sigma)$ may be computed in
$O(n^{\lexp{d}}+n\log{}n+T_{d+1}(n_1,\ldots,n_m))$ time, where
$T_{d+1}(n_1,\ldots,n_m)$ stands for the time needed to compute the
convex hull of $m$ $d$-polytopes lying on $m$ parallel hyperplanes in
$\RR^{d+1}$, where the $i$-th polytope has $n_i$ vertices (cf. Section
\ref{sec:chs-algo}).

Finally, Boissonnat and Karavelas \cite{bk-ccevc-03} have
shown that convex hulls of spheres in $\RR^d$ and additively
weighted Voronoi cells in $\RR^d$ are combinatorially
equivalent. This equivalence suggests that we should be able to
refine the worst-case complexity of an additively
weighted Voronoi cell in any odd dimension, when the number of distinct
radii of the spheres involved is considered fixed.